\documentclass[journal]{IEEEtran}
\usepackage{hyperref}
\usepackage{color}
\usepackage{amsmath}
\usepackage{bm}
\usepackage{amssymb}
\usepackage{graphicx}
\usepackage{amsthm}
\usepackage{cite}
\usepackage[ruled,linesnumbered]{algorithm2e}
\usepackage{multirow}
\usepackage{diagbox}
\usepackage{subfigure}
\usepackage{hyperref}
\usepackage[figure]{hypcap}
\usepackage{multicol}
\usepackage{cuted} 
\usepackage{balance}  
\usepackage[table]{xcolor}    
\usepackage{booktabs}    
\usepackage{tabularx}
\usepackage{makecell}
\usepackage{array}            
\hypersetup{
  colorlinks=true,
  linkcolor=blue,
  citecolor=blue,
  urlcolor=blue
}
\usepackage{cleveref}
\usepackage{upgreek}
\usepackage{threeparttable} 

\newtheorem{proposition}{Proposition}

\newcolumntype{Y}{>{\centering\arraybackslash}X}

\begin{document}

\title{An Adaptive Antenna Impedance Matching \\Method via Deep Reinforcement Learning
}

\author{Guoquan Zhang, Wendong Cheng, Weidong Wang, and Li Chen 
\thanks{Guoquan Zhang, Wendong Cheng, Weidong Wang, and Li Chen are with the CAS Key Laboratory of Wireless Optical Communication, University of Science and Technology of China (USTC), Hefei 230027, China (e-mail: zgq2002@mail.ustc.edu.cn; cwd01@mail.ustc.edu.cn; wdwang@ustc.edu.cn; chenli87@ustc.edu.cn;).}
}



\maketitle

\begin{abstract}
Adaptive impedance matching between antennas and radio frequency front-end modules is critical for maximizing power transmission efficiency in mobile communication systems. Conventional numerical and analytical methods struggle with a trade-off between accuracy and efficiency, while deep neural network (DNN)-based supervised learning approaches rely heavily on large labeled datasets and lack flexibility for dynamic environments. To address these limitations, this paper proposes a deep reinforcement learning (DRL)-based approach for adaptive impedance matching. First, we model the impedance tuning problem as an optimal control problem, proving the feasibility of solving the optimal control law via reinforcement learning. Then, we design a tailored DRL framework for impedance tuning, which employs a compact state representation that integrates key frequency characteristics and matching quality metrics.
Additionally, this framework incorporates a piecewise reward function that accounts for both matching accuracy and tuning speed. Furthermore, a test-phase exploration mechanism is introduced to enhance tuning stability, which effectively reduces local optimal trapping and high-frequency tuning variance. Experimental results demonstrate that the proposed method achieves superior performance in terms of tuning accuracy, efficiency, and stability compared with conventional heuristic and gradient-based methods, making it promising for practical impedance tuning systems.
\end{abstract}
\section{Introduction}
\IEEEPARstart{I}{mpedance} matching is a crucial technology in radio frequency (RF) circuits, aiming to maximize power transfer efficiency \cite{intro,rf,rf2}. In mobile communication systems, impedance mismatch between the antenna and RF front-end (RFFE) degrades signal quality, shortens battery life, impairs power amplifier linearity \cite{zenteno2015output, Wangxiaoyu2022Digital}, and may even damage sensitive RFFE components \cite{van2007power}. Therefore, impedance matching is indispensable for reliable, high-performance operation in modern mobile devices.

Moreover, the antenna impedance in mobile devices is inherently dynamic, affected by a multitude of real-world factors: operating frequency \cite{alibakhshikenari2019automated}, variations in user holding postures \cite{boyle2003performance, ogawa2001analysis}, user proximity effects \cite{boyle2007analysis, adams2023miniaturized}, and even user age and clothing \cite{sacco2021antenna}. These dynamic conditions induce persistent impedance mismatches, which reduce the power delivered to the antenna and threaten the long-term reliability of the
entire communication system. Given this dynamic operating environment, adaptive impedance matching techniques have emerged as a critical research focus.

For adaptive impedance matching, conventional analytical methods obtain the optimal matching parameters through theoretical derivation based on the circuit structure and actual impedance measurements.
To address antenna mismatch in mobile phones caused by fluctuating body effects, the authors in \cite{lnetwork} developed a generic quadrature detector to achieve power-independent orthogonal measurement of complex impedance, enabling the direct adjustment of tunable capacitors. Additionally, the work of \cite{match} proposed an analytical method that directly computes the optimal component values of the matching network based on the measured load impedance and circuit model. Further, the authors in \cite{analytical} proposed a matching method for $\pi$-network impedance tuners, which uses closed-form formulas to achieve impedance matching within finite tuning ranges. Analytical methods avoid tedious iterative searches, achieving high matching efficiency. As analytical
methods are inherently model-dependent, their accuracy is limited by discrepancies between the assumed circuit model and the actual physical system.

To overcome these model-dependent limitations, numerical iterative optimization methods have been widely adopted for adaptive impedance matching. These methods utilize real-time feedback signals indicating the level of mismatch to search for optimal matching parameters through iterative adjustments. Gradient descent algorithm utilizes the gradient information to drive stepwise parameter updates \cite{Xiong2016Unimodal, ogawa2003automatic}, yet it often suffers from slow convergence and is prone to stagnation in local optima. To eliminate the reliance on gradient information, several gradient-free optimization methods, including the Powell algorithm and the Single-step algorithm, are adopted in \cite{de2004rf} to minimize the reflection coefficient magnitude. As a major category of gradient-free techniques, heuristic methods are widely employed for impedance matching, as they iteratively search for optimal matching parameters through intelligent heuristic strategies. Typical examples include genetic algorithms (GA) \cite{refga} and its variant \cite{refgaimproved}, which exhibit high complexity due to the procedures of selection, crossover and mutation. In contrast, particle swarm optimization (PSO) \cite{pso} is much simpler than GA as it lacks these genetic operations. To alleviate premature convergence and local optimum trapping, a simulated annealing particle swarm optimization (SAPSO) algorithm was proposed in \cite{sapso} for impedance matching, which incorporates the simulated annealing (SA) mechanism into the PSO framework. In addition, to accelerate convergence speed and reduce the hardware cost of feedback circuits, the authors in \cite{xiong2019novel} proposed a binary search tuning scheme based on linear fractional transformation. The drawback of numerical iteration methods lies in their inherent inefficiency, as the trial-and-error search process incurs significant tuning latency and computational overhead.

Recently, advanced artificial intelligence (AI) techniques have been applied to achieve efficient and accurate adaptive impedance matching. For frequency-domain impedance mismatch, Kim and Bang \cite{match_dnn} developed a deep neural network (DNN) that directly predicts the component values of an L-type matching network from only the magnitude of the reflection coefficient, thus avoiding complex impedance measurement and iterative tuning procedures. Similarly, a low-complexity, hidden-layer-free shallow learning model was presented in \cite{slm}, which can determine the component values of matching circuits in real time solely using the magnitude of antenna reflection coefficients. In addition, Jeong et al. \cite{jeong2019real} introduced a real-time range-adaptive impedance matching method for wireless power transfer systems using neural network-based machine learning. Furthermore, Cheng et al. \cite{chengwendong} proposed a DNN that directly outputs the optimal $\pi$-network matching solution for time-frequency domain impedance matching, with frequency, voltage standing-wave ratio (VSWR), and peak voltages as inputs. To achieve real-time impedance matching for variable loads in RF systems, a deep learning-based state transfer adaptive matching network architecture was developed in \cite{wangkun2025sate}, which integrates non-invasive voltage and current probes with the DNN. To alleviate impedance mismatch under parasitic effects, Cheng et al. \cite{cheng2025data-driven} presented a data-driven adaptive impedance matching scheme. This scheme employs a residual-enhanced neural network to characterize unknown S-parameters influenced by parasitic effects, and adopts an inverse mapping network to rapidly and accurately predict the optimal parameters. Overall, these DNN-based impedance matching methods exhibit excellent performance in terms of both accuracy and efficiency.

However, the DNN-based methods require substantial labeled data and lack sufficient flexibility in the face of dynamic environmental changes. Reinforcement learning (RL), which learns online through real-time interaction with the environment, is naturally suitable for dynamic impedance tuning tasks without requiring large amounts of pre‑labeled data. 
To the best of our knowledge, few existing works have applied RL in a specific impedance tuning system. Despite its promising potential, two critical challenges remain for RL-based impedance tuning. First, due to the lack of existing applications, the theoretical foundation for applying RL to impedance tuning is insufficiently clarified, leaving the rationality of its application uncertain. Second, the design of core RL elements, especially the state and reward function, remains a major challenge, as they must fit the specific impedance tuning task and enable the agent to learn a convergent policy for fast and accurate impedance matching. To address these challenges, in this paper, we first model the impedance tuning problem as an optimal control problem and verify that the optimal control law can be solved via RL. We then propose a deep reinforcement learning (DRL)-based adaptive impedance matching method, with key DRL elements carefully designed for the impedance matching task. Finally, we compare the matching performance of the proposed method with conventional heuristic algorithms and a gradient-based method in terms of matching accuracy and efficiency. To summarize, our contributions are as follows.
\begin{itemize}
    \item \textit{Control-Theoretic Modeling for Impedance Tuning:} We formulate the impedance tuning problem as an optimal control problem using the state-space method, including the definition of system state, control input, and state evolution equation. Furthermore, we prove the feasibility of solving the optimal control law for this problem via RL, by analyzing the reward function and the action-value function. This formulation bridges the gap between adaptive impedance matching and optimal control theory, establishing a solid theoretical foundation for the proposed DRL solution.

    \item \textit{Tailored DRL-based Method for Adaptive Impedance Tuning:} We propose a tailored DRL-based method for impedance tuning, which trains a DRL agent to learn an approximate optimal control law via online interactions with the environment. Specifically, we design a state representation that integrates both matching quality metrics and frequency indicators to accurately characterize the system state. Furthermore, to effectively incentivize the agent to learn excellent policies, we propose a carefully designed piecewise reward function that consists of three components: a base reward, an improvement reward, and a fast convergence reward. The proposed DRL-based matching method achieves fast and accurate impedance matching without relying on massive pre-labeled data.

    \item \textit{Robustness Enhancement for Tuning Stability:} We introduce an effective test-phase exploration mechanism to enhance the tuning stability of the DRL-based impedance matching method. By maintaining a small, controlled exploration rate during test-phase inference, the agent can effectively escape from local optima and significantly reduce tuning variance at high frequencies. This leads to more consistent and reliable performance over the operating frequency band, alleviating the common issues of high-frequency fluctuations and local optima trapping, which is critical for practical deployment.
\end{itemize}

The rest of this paper is organized as follows. Section \ref{sec:model} introduces an adaptive impedance matching system. 
Section \ref{sec:RL_control} 
analyzes the system from a control-theoretic perspective and proves the feasibility of RL-based optimal control law solution. Section \ref{sec:DRL_algorithm} details the proposed DRL-based impedance matching method. Section \ref{sec:numerical_results} presents comprehensive experimental results and performance analyses. Finally, Section \ref{sec:conclusion} concludes the paper.

\section{System Model}
\label{sec:model}
In this section, we first introduce an adaptive impedance matching system, followed by an analysis of the closed-form solution for the ideal L-network, and finally model the impedance tuning process.

A typical adaptive impedance matching system, as shown in Fig. \ref{fig:system}, comprises three key modules: a tunable matching network (TMN), an impedance sensor, and a tuner control unit. The TMN, typically configured as an L-network or $\pi$-network with tunable capacitors and inductors, performs impedance transformation by adjusting its reactive components. A bi-directional coupler extracts incident and reflected signals, from which the impedance sensor detects impedance variations and derives the reflection coefficient or VSWR. Based on these measurements, the tuner control unit determines the appropriate adjustments to the TMN using an adaptive impedance matching method. In this work, we focus on the scenario where the available measurement data is accurate and complete complex impedance information (e.g., the reflection coefficient), rather than merely scalar values, to enable precise and efficient impedance matching. 
\begin{figure}
\centering
\includegraphics[scale=0.5]{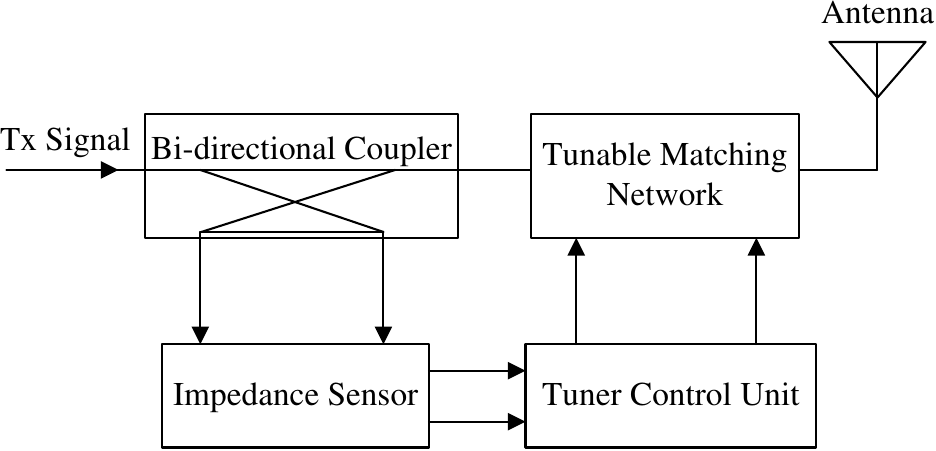}
\caption{Block diagram of an adaptive impedance matching system.}\label{fig:system} 
\end{figure}

Without loss of generality, we adopt the L-network as the TMN configuration for our analysis. As shown in Fig. \ref{fig:circuit}, the L-network uses two tunable capacitors $C_p$ and $C_s$ which enable the impedance transformation. The input impedance $Z_\text{in}$ after impedance transformation by the
L-network is given by
\begin{equation}
Z_\text{in}=\frac{1}{jB_p+\dfrac{1}{Z_L+\dfrac{1}{jB_s}}},
\label{eq:Zin}
\end{equation}
where $B_p=\omega C_p$ and $B_s=\omega C_s$ are the susceptances of $C_p$
and $C_s$, respectively, $Z_L$ denotes the load impedance, and $\omega$
represents the angular frequency.

Based on Eq. \eqref{eq:Zin}, the reflection coefficient which characterizes the input impedance can be expressed as
\begin{equation}
    \Gamma_\text{in}=\frac{Z_\text{in}-R_s}{Z_\text{in}+R_s},
    \label{eq:Gamma_in}
\end{equation}
where $R_s$ represents the source impedance.
For maximum power transfer, the input impedance $Z_\text{in}$ must equal the complex conjugate of the source impedance. As the standard source impedance in RF systems is typically a purely resistive 50 $\Omega$, this matching condition is reduced to $Z_\text{in}=R_s$, which consequently implies that the reflection coefficient $\Gamma_\text{in}$ is zero. The TMN fulfills this requirement by tuning the values of $C_p$ and $C_s$. Denoting the load impedance as $Z_L = R_L + jX_L$ and substituting it into Eq. \eqref{eq:Zin}, we get the
conjugate matching equations as
\begin{equation}
    \begin{cases}
\dfrac{B_s^2R_L}{\left(B_pB_sR_L\right)^2+\left(B_p+B_s-B_pB_sX_L\right)^2}=R_S \\[10pt] 
\dfrac{2B_pB_sX_L+B_s^2X_L-B_pB_s^2\left(R_L^2+X_L^2\right)-B_p-B_s}{\left(B_pB_sR_L\right)^2+\left(B_p+B_s-B_pB_sX_L\right)^2}=0.  
\end{cases}
\label{eq:conjugate matching}
\end{equation}

By solving Eq. \eqref{eq:conjugate matching}, we obtain the closed-form expressions of the capacitor values required for impedance matching as
\begin{equation}
    \begin{cases}
C_p=\dfrac{R_LR_S-R_L^2\pm X_L\sqrt{R_LR_S-R_L^2}}{\omega\left[R_LX_LR_S\pm R_LR_S\sqrt{R_LR_S-R_L^2}\right]} \\[10pt]
C_s=\dfrac{X_L\pm\sqrt{R_LR_S-R_L^2}}{\omega\left[R_L^2+X_L^2-R_LR_S\right]}. 
\end{cases}
\label{eq:CpCs}
\end{equation}

\begin{figure}
\centering
\includegraphics[scale=0.8]{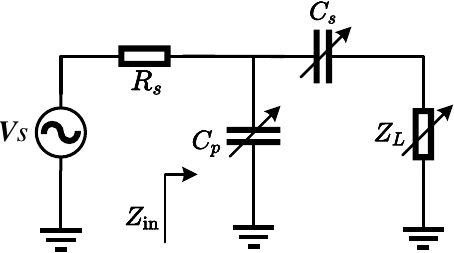}
\caption{The architecture of the L-network.}
\label{fig:circuit}
\end{figure}

Despite the availability of closed-form expressions in Eq. \eqref{eq:CpCs} for optimal capacitances, their direct application in practice is infeasible. This is due to the inherent model inaccuracies, component tolerances, and non-negligible parasitic effects in real RF circuits \cite{cheng2025data-driven}, which deviate from the ideal circuit.  Additionally, the load impedance (e.g., of an antenna) is often unknown or dynamic in practice \cite{firrao2008automatic}, as it exhibits frequency-dependent and environment-sensitive characteristics.

In practical scenarios, impedance matching is typically achieved through a dynamic impedance tuning process, where an algorithm iteratively explores and optimizes the matching parameters. 
To characterize this tuning process, we first establish the system model. Let \( \boldsymbol{\alpha} = [\alpha_1, \alpha_2, \cdots, \alpha_n]^T \) denote the tunable parameters of the TMN, \( Z_{\text{in}}(\boldsymbol{\alpha}) \) be the input impedance determined by \( \boldsymbol{\alpha} \), and \( \Gamma_{\text{in}}(\boldsymbol{\alpha}) = \frac{Z_{\text{in}}(\boldsymbol{\alpha}) - R_s}{Z_{\text{in}}(\boldsymbol{\alpha}) + R_s} \) be the reflection coefficient.
At tuning step \( k \), the initial parameter vector is \( \boldsymbol{\alpha}_k \), and the corresponding reflection coefficient magnitude is \( |\Gamma_{\text{in}}(\boldsymbol{\alpha}_k)| \). 
The algorithm updates \( \boldsymbol{\alpha}_k \) to \( \boldsymbol{\alpha}_{k+1} = \boldsymbol{\alpha}_k + \Delta \boldsymbol{\alpha} \), where \( \Delta \boldsymbol{\alpha} \) is the parameter adjustment, such that \( |\Gamma_{\text{in}}(\boldsymbol{\alpha}_{k+1})| < |\Gamma_{\text{in}}(\boldsymbol{\alpha}_k)| \), thereby reducing the degree of mismatch. 
The tuning process terminates when a stopping criterion is satisfied (e.g., $|\Gamma_\text{in}|$ falls below a threshold), and the total number of tuning steps $K$ is adaptive and determined by the algorithm.
Mathematically, the objective of tuning is to find the optimal sequence of parameter adjustments \( \{\Delta \boldsymbol{\alpha}\}_{k=1}^K \) that minimizes $|\Gamma_\text{in}|$ at each step, which is formulated as
\begin{equation}
\begin{aligned}
 \Delta\boldsymbol{\alpha}_k &= \arg\min_{\Delta\boldsymbol{\alpha}_k} \left| \Gamma_{\text{in}}(\boldsymbol{\alpha}_{k-1} + \Delta\boldsymbol{\alpha}_k) \right|, \quad k = 1, 2, \dots, K \\
& \textit{s.t.} \quad \left| \Gamma_{\text{in}}(\boldsymbol{\alpha}_K) \right| \leq \Gamma_{\text{th}}
\end{aligned}
\end{equation}
where $\boldsymbol{\alpha}_0$ is the initial parameter of the TMN, and $\Gamma_\text{th}$ denotes the predefined reflection coefficient magnitude threshold.

However, conventional tuning methods rely on exhaustive trial-and-error searches, leading to low efficiency and slow convergence. To overcome these limitations and develop a high-performance impedance tuning approach, we analyze the tuning system from a control-theoretic perspective in subsequent sections. Based on optimal control theory, we derive the optimal control law for impedance tuning, and then employ DRL to approximate this law in a data-driven manner.

\section{From Traditional Optimal Control to RL for Impedance Tuning}
\label{sec:RL_control}
In this section, we first introduce the fundamentals of optimal control. Building on this theoretical basis, we model the impedance tuning system as a control system. Subsequently, we derive the optimal control law for the impedance tuning system and establish its theoretical connection to RL.

\subsection{Basics of Optimal Control}
\label{subsec:Basics of Optimal Control}
To formulate the impedance tuning system within an optimal control framework, we first introduce the fundamental elements of state-space analysis, followed by a brief derivation of the optimal control law.

1) \textit{State Vector:}
The state vector serves as a fundamental descriptor characterizing the dynamic behavior of a control system, as it encapsulates all necessary information from the past to uniquely determine its future evolution. In its general form, it is represented as an $n$-dimensional column vector
\begin{equation}
\bm{x}=\left[x_1,x_2,\cdots,x_n\right]^T.
\end{equation}

2) \textit{Control Input:}
The control input represents the manipulable variable that governs the state transition and determines the dynamic evolution of the control system. Similarly, the control input is represented as an $m$-dimensional column vector
\begin{equation}
\bm{u}=\left[u_1,u_2,\cdots,u_m\right]^T.
\end{equation}

3) \textit{State Evolution Equation:}
The relationship between the state vector, control input, and system dynamics is described by the state evolution equation.
For a discrete-time system, which is sampled at discrete time instants $k$, the state transition is expressed as a difference equation
\begin{equation}
    \boldsymbol{x}_{k+1} = \boldsymbol{f}\left(\bm{x}_{k},\boldsymbol{u}_{k},k\right),
\end {equation}
where $\bm{x}_k$ and $\bm{u}_k$ represent the state and control input at the $k$-th time step, respectively.

4) \textit{Control Performance Metric:} 
While the continuous-time linear quadratic (LQ) performance metric is widely adopted in theoretical analysis \cite{anderson2007optimal}, we focus on its discrete-time counterpart here, as it aligns directly with the digital implementation of our impedance tuning system. The discrete-time LQ metric is defined as
\begin{equation}
    J = \sum_{k=0}^{\infty} \gamma^k\left( \boldsymbol{x}^T_k \boldsymbol{Q} \boldsymbol{x}_k + \boldsymbol{u}^T_k \boldsymbol{R} \boldsymbol{u}_k \right),
\end{equation}
where $\gamma \in (0,1)$ is the discount factor to ensure the convergence of the metric $J$, $\bm{Q} \succeq 0$ and $\bm{R} \succ 0$ are the weighting matrices for the state and control input, respectively. This metric $J$ is also referred to
as the cost function, which balances the control accuracy and control consumption.

5) \textit{Optimal Control Law and Bellman Equation:} 
The control objective is to find the optimal control law \(\boldsymbol{u}^* = \phi(\boldsymbol{x})\) that minimizes the cost \(J\). To this end, we first define the optimal value function $V^*(\boldsymbol{x}_k)$ as 
\begin{equation}
V^*(\boldsymbol{x}_k) = \min_{\boldsymbol{u}_k, \boldsymbol{u}_{k+1}, \ldots} \sum_{i=k}^{\infty} \gamma^{i-k} \left( \boldsymbol{x}_i^T \boldsymbol{Q} \boldsymbol{x}_i + \boldsymbol{u}_i^T \boldsymbol{R} \boldsymbol{u}_i \right).
\label{eq:Vx*}
\end{equation}
which represents the minimal cumulative cost-to-go starting from state $\bm{x}_k$ under the optimal control law.

Based on the dynamic programming principle \cite{bertsekas2012dynamic}, we split the infinite-horizon cost sum into the immediate cost and the discounted future cost, yielding the Bellman optimality equation
\begin{equation}
    V^*(\boldsymbol{x}_k) = \min_{\boldsymbol{u}_k} \left\{ \left( \boldsymbol{x}_k^T\boldsymbol{Q}\boldsymbol{x}_k + \boldsymbol{u}_k^T\boldsymbol{R}\boldsymbol{u}_k \right) + \gamma V^*(\boldsymbol{x}_{k+1}) \right\},
    \label{eq:dis_bellman}
\end{equation}
where $\boldsymbol{x}_{k+1} = \boldsymbol{f}\left(\bm{x}_{k},\boldsymbol{u}_{k}\right)$ is the discrete-time state evolution equation. This equation directly provides the optimal control input $\boldsymbol{u}^*_k$ for the current state $\boldsymbol{x}_k$, which implicitly represents the optimal control law $\boldsymbol{u}_k^* = \phi(\boldsymbol{x}_k), (k=0,1,2,\cdots)$. For simplicity, we define a discrete Q-function $Q^*_H(\cdot)$ as 
\begin{equation}
\label{eq:QH_func}Q^*_H(\boldsymbol{x}_k,\boldsymbol{u}_k) =  \boldsymbol{x}_k^T\boldsymbol{Q}\boldsymbol{x}_k + \boldsymbol{u}_k^T\boldsymbol{R}\boldsymbol{u}_k + \gamma V^*(\boldsymbol{x}_{k+1}).
\end{equation}
Thus, the optimal control law is finally expressed as
\begin{equation}
     \label{eq: optimal control law min}
    \boldsymbol{u}^*_k=\phi (\boldsymbol{x}_k)= \arg \min_{\boldsymbol{u}_k} Q^*_H(\boldsymbol{x}_k, \boldsymbol{u}_k).
\end{equation}

\subsection{Control-Theoretic Modeling of the Impedance Tuning System}
In this part, we model the impedance tuning system as an optimal control system within the state-space framework. Fig. \ref{fig:tuning_control_system} illustrates the block diagram of the impedance tuning system, where the physical quantities (e.g., reflection coefficient \(\Gamma_{\text{in}}\) and capacitance adjustment) are clearly depicted. 
\begin{figure}
\centering
\includegraphics[scale=0.45]{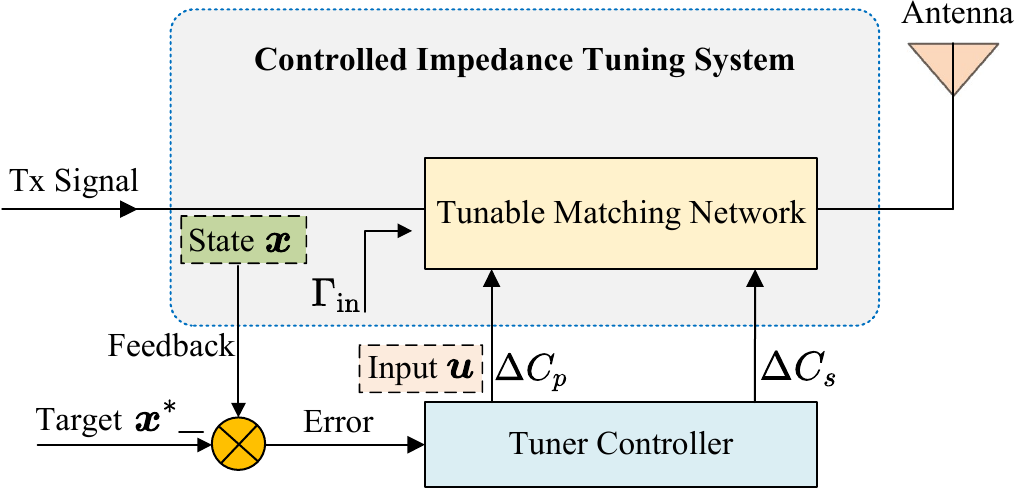}
\caption{Block diagram of the impedance tuning system under the control-theoretic framework.}\label{fig:tuning_control_system} 
\end{figure}
Specifically, we map these physical quantities to the corresponding control-theoretic variables as follows:

1) \textit{State Vector:}
The reflection coefficient $\Gamma_\text{in}$ in Eq. \eqref{eq:Gamma_in} is selected as the core state variable, which can be directly measured by the bi-directional coupler and impedance sensor. By decomposing $\Gamma_\text{in}$ into its real and imaginary parts, a two-dimensional state vector is constructed as
\begin{equation}
    \boldsymbol{x} = \begin{bmatrix} x_1 \\x_2 \end{bmatrix} = \begin{bmatrix} \text{Re}(\Gamma_\text{in}) \\ \text{Im}(\Gamma_\text{in})\end{bmatrix},
\end{equation}
where $x_1=\mathrm{Re}(\Gamma_\text{in})$, $x_2=\mathrm{Im}(\Gamma_\text{in})$.

The core goal of impedance tuning is to drive the state vector $\boldsymbol{x}$ to asymptotically converge to the target equilibrium state $\boldsymbol{x}^*=[0,0]^T$. When this convergence is achieved, the reflection coefficient satisfies \(|\Gamma_\text{in}| = \sqrt{x_1^2 + x_2^2} = 0\), which means zero power reflection between the source and the matching network, thereby realizing perfect impedance matching.

2) \textit{Control Input:}
The incremental adjustments of capacitance are adopted as control inputs, which better align with the practical tuning scenario: the parameters of TMN are adjusted incrementally based on their current values. The control input vector is defined as
\begin{equation}
    \boldsymbol{u} = \begin{bmatrix} u_1 \\ u_2 \end{bmatrix} = \begin{bmatrix} \Delta{C_p} \\ \Delta{C_s} \end{bmatrix},
\end{equation}
where \(u_1=\Delta{C_p} \) and \(u_2=\Delta{C_s}\) denote the incremental adjustments of the capacitors $C_p$ and $C_s$, respectively. 

Given that the two capacitances are constrained to stay within their physically allowable ranges and their initial values \(C_p(0)\) and \(C_s(0)\), the current capacitance values can be expressed as
\begin{equation}
\begin{cases}
C_p(k) = C_p(0) + \sum_{i=0}^{k-1} \Delta C_p(i) \\
C_s(k) = C_s(0) + \sum_{i=0}^{k-1} \Delta C_s(i)
\end{cases},
\end{equation}
where \(C_p(k)\) and \(C_s(k)\) are the capacitance values at the \(k\)-th time step. The feasible control input set \(\mathcal{U}\) is thus defined as
\begin{equation}
\boldsymbol{u} \in \mathcal{U} = \left\{ (\Delta{C_p},\Delta{C_s}) \middle| \quad
\begin{aligned}
C_{p,\min} &\leq C_p(k) \leq C_{p,\max}, \\
C_{s,\min} &\leq C_s(k) \leq C_{s,\max}
\end{aligned}
\right\},
\end{equation}
where $C_{p,\min}$ and $C_{p,\max}$ denote the lower and upper limits of $C_p$, while $C_{s,\min}$ and $C_{s,\max}$ denote those of $C_s$.

3) \textit{State Evolution Equation:}
The dynamic state evolution of the impedance tuning control system is determined by the control input $\boldsymbol{u}_k$, and its state equation can be expressed as
\begin{equation}
    \label{eq:state equation}
    \boldsymbol{x}_{k+1} = \boldsymbol{f}\left(\bm{x}_k,\boldsymbol{u}_k\right) = \begin{bmatrix}
        \mathrm{Re} (\Gamma_{\text{in},k+1})\\
        \mathrm{Im}(\Gamma_{\text{in},k+1})
    \end{bmatrix},
\end{equation}
where $\Gamma_{\text{in},k+1}$ denotes the $\Gamma_\text{in}$ at the $(k+1)$-th time step. The function $\boldsymbol{f}(\cdot)$ represents a nonlinear vector-valued mapping whose closed-form expression is not analytically tractable. This intractability arises from three main factors:
\begin{itemize}
    \item The nonlinear mapping relationship between the input impedance $Z_\text{in}$ and capacitances $C_p$, $C_s$ induces a nonlinear correlation between $\Gamma_\text{in}$ and the control input $\boldsymbol{u}$.
    \item Deriving the expression requires extracting real and imaginary parts of complex-valued $\Gamma_\text{in}$, making the resultant formula difficult to simplify and excessively cumbersome.
    \item Parasitic effects inherent in practical systems further complicate the system model, rendering an exact closed-form representation of $\bm{f}(\cdot)$ infeasible.
\end{itemize}
\subsection{Optimal Control Law for the Impedance Tuning System}
\label{subsec:Optimal_Control_Law_Tuning}
Building on the control-theoretic model established in the preceding subsections, the optimal control law for the impedance tuning system can theoretically be obtained by solving the following set of equations
\begin{equation}
\label{eqs:the_optimal_law_expr}
    \begin{cases}
        \bm{u}_k=\arg \min_{\bm{u}_k} \{\boldsymbol{x}_k^T\boldsymbol{Q}\boldsymbol{x}_k + \boldsymbol{u}_k^T\boldsymbol{R}\boldsymbol{u}_k + \gamma V^*(\boldsymbol{x}_{k+1})  \}
        \\
        \bm{x}_{k+1}=\bm{f}(\bm{x}_k,\bm{u}_k)
    \end{cases}.
\end{equation}

While analytical solutions to Eq. \eqref{eqs:the_optimal_law_expr} exist only for
simple linear systems such as linear quadratic regulator (LQR) systems \cite{Terra2014Optimal}, solving the equations for our nonlinear
impedance tuning system is analytically intractable. This difficulty stems from the highly complex state evolution equation without an explicit form, which makes conventional analytical methods infeasible. To handle such nonlinearity and avoid intractable analytical derivations, we introduce a RL-based solution framework to obtain the optimal control law.
First, we define the key elements of the RL framework, and then derive the connection between the optimal control law and RL. 
We define the reward function in RL as
\begin{equation}
r(\boldsymbol{x}, \boldsymbol{u}) = - \left( \boldsymbol{x}^T Q \boldsymbol{x} + \boldsymbol{u}^T R \boldsymbol{u} \right).
\label{eq:reward_definition}
\end{equation}
The optimal action-value function (optimal Q-function) in RL is defined as
\begin{equation}
Q^*_\text{RL}(\boldsymbol{x}_k, \boldsymbol{u}_k) = \min_{\bm{u}_{k+1},\bm{u}_{k+2},\ldots} \sum_{i=k}^{\infty} \gamma^{i-k} r(\boldsymbol{x}_i, \boldsymbol{u}_i),
\label{eq:rl_q}
\end{equation}
which represents the minimal cumulative reward starting from state $\bm{x}_k$ with given action $\bm{u}_k$. Based on these definitions, the equivalence between optimal Q-function in RL and discrete Q-function in \eqref{eq:QH_func} can be summarized in the following proposition.
\begin{proposition}[Q-Function Equivalence]
    \label{pro:Q-Function Equivalence}
The optimal Q-function in RL and the discrete Q-function in \eqref{eq:QH_func} satisfy
\begin{equation}
    \label{eq:Q-function equivalence}
    Q^*_\text{RL}(\boldsymbol{x}_k,\boldsymbol{u}_k) = -Q^*_{H}(\boldsymbol{x}_k,\boldsymbol{u}_k).
\end{equation}
\end{proposition}
\begin{IEEEproof}
We expand $Q^*_\text{RL}(\bm{x}_k,\bm{u}_k)$ according to its definition as follows
\begin{align}
Q^*_\text{RL}(\bm{x}_k, \bm{u}_k) 
&= r(\bm{x}_k, \bm{u}_k) + \gamma \min \sum_{i=k+1}^{\infty} \gamma^{i-(k+1)} r(\bm{x}_i, \bm{u}_i) \notag \\
&= -\left( \bm{x}_k^T Q \bm{x}_k + \bm{u}_k^T R \bm{u}_k \right) \notag \\
&\quad  - \gamma \min \sum_{i=k+1}^{\infty} \gamma^{i-(k+1)} \left( \bm{x}_i^T Q \bm{x}_i + \bm{u}_i^T R \bm{u}_i \right) \notag \\
&= -\left[ \bm{x}_k^T Q \bm{x}_k + \bm{u}_k^T R \bm{u}_k  + \gamma V^*(\bm{x}_{k+1}) \right] \notag \tag{by Eq. \eqref{eq:Vx*}} \\
&= -Q^*_H(\bm{x}_k, \bm{u}_k) .
\end{align}
This completes the proof of Proposition \ref{pro:Q-Function Equivalence}.
\end{IEEEproof}
Therefore, the optimal control law can be directly obtained through the Q-function in RL as follows
\begin{equation}
    \bm{u}^*=\phi (\bm{x}) = \arg\max_{\bm{u}\in\mathcal{U}} Q^*_\text{RL}(\bm{x},\bm{u}).
\end{equation}

In summary, our derivation confirms that the optimal control law for the impedance tuning system can be directly obtained by maximizing the optimal Q-function  $Q^*_\text{RL}(\bm{x},\bm{u})$ with respect to $\bm{u}$ in RL, providing a theoretical basis for subsequent algorithm design.

\section{DRL-Based Impedance Tuning Algorithm}
\label{sec:DRL_algorithm}
In this section, we first introduce the use of DRL to approximate the optimal control law for impedance tuning, and then propose a DRL-based impedance tuning algorithm.

\subsection{Basics of Deep Reinforcement Learning}
To facilitate the presentation of our design, we briefly introduce some key concepts of DRL in this subsection.

Markov decision process (MDP) is the foundational mathematical framework for RL model. An MDP is formally defined by the tuple  $\langle \mathcal{S},\mathcal{A},P,R\rangle$, where $\mathcal{S}$ denotes the state space, $\mathcal{A}$ represents the action space, $P$ defines the state transition probability, $R$ is the immediate reward. When an agent in state $s \in \mathcal{S}$ executes an action $a \in \mathcal{A}$, the environment transitions to a next state $s^{\prime} \in \mathcal{S}$ with a probability given by $P(s^\prime | s,a)=Pr(S_{t+1}=s^\prime | S_{t}=s, A_t=a)$. Concurrently, the agent receives an immediate reward 
$R(s,a)$. 
The agent’s action is governed by a policy 
\begin{equation}
    \pi (a \mid s)=Pr(A_t=a \mid S_t =s).
\end{equation}
which maps states to a probability distribution over actions. The expected cumulative reward is defined as the return, whose expression is given by
\begin{equation}
    G_t = R_t + \gamma R_{t+1} + \gamma^2 R_{t+2} + \cdots = \sum_{k=0}^\infty \gamma^k R_{t+k},
\end{equation}
where $\gamma \in [0,1)$ is the discount factor for future rewards. Due to the inherent stochasticity in both environment transitions and the policy itself, the return $G_t$ is a random variable. Consequently, the core optimization problem is formulated as 
\begin{equation}
    \max_\pi  \mathbb{E}(G_t).
\end{equation}
The objective of RL is seeking the policy $\pi(a | s)$ that yields the highest expected return.

The definition of the action-value function in RL is given as follows
\begin{equation}
    Q^\pi(s, a) = \mathbb{E}_\pi\left[G_t \mid S_t=s, A_t=a\right],
\end{equation}
which is the conditional expected return for an agent to select action $a$ in the state $s$ under the policy $\pi$. For any policy $\pi$ and any state $s \in \mathcal{S}$, action-value function satisfies the following recursive relationship
\begin{align}
& Q^{\pi}(s, a) \notag \\
 &= \mathbb{E}_{s^{\prime}} \left[ R(s,a) + \gamma \sum_{a' \in \mathcal{A}} \pi(a'|s') Q^{\pi}(s', a') \middle| s, a \right]  \notag \\
 &= \sum_{s' \in \mathcal{S}} P(s'|s, a) \left(R(s,a) + \gamma \sum_{a' \in \mathcal{A}} \pi(a'|s') Q^{\pi}(s', a') \right) ,
\label{eq:q_bellman}
\end{align}
where  $R(s,a)$ is the immediate reward when the environment transits from state $s$ to state $s^\prime$ after taking the action $a$, and Eq. \eqref{eq:q_bellman} is the well-known Bellman equation of action-value function \cite{sutton1998reinforcement}.

A policy is deemed superior to another if its expected return outperforms that of the alternative across all possible states and actions. On this basis, the optimal action-value function can be expressed as
\begin{equation}
\label{eq:optimal_Q_expr}
    Q^*(s,a)=\max_\pi Q^\pi (s,a), \quad \forall s \in \mathcal{S},\ a \in \mathcal{A}
\end{equation}
Given the optimal action-value function, the corresponding optimal policy is uniquely determined as
\begin{equation}
\label{eq:policy_expr}
    \pi^*(a|s) =
\begin{cases}
1, & \text{if } a = \arg \max_{a \in \mathcal{A}} Q^*(s, a) \\
0, & \text{otherwise}
\end{cases}
\end{equation}
By integrating Eqs. \eqref{eq:optimal_Q_expr}, \eqref{eq:policy_expr} with \eqref{eq:q_bellman}, the Bellman optimality equation for $Q^*(s,a)$ is given by
\begin{align}
Q^*(s, a)
&= \sum_{s' \in \mathcal{S}} P(s'|s, a) \left( R(s,a) + \gamma \max_{a' \in \mathcal{A}} Q(s', a') \right) \notag \\
&= \mathbb{E}_{s'} [ R(s,a) + \gamma \max_{a' \in \mathcal{A}} Q^*(s', a') \mid s_t = s,, a_t = a ].
\end{align}

Based on the Bellman optimality equation, we can derive the optimal policy $\pi^*(a | s)$ and the corresponding $Q^*(s,a)$ using iterative techniques, such as policy iteration and value iteration algorithms \cite{sutton1998reinforcement}. In
the following discussion, our focus will be placed on value
iteration-based approaches. 

When both the state $s$ and action $a$ are discrete, the $Q^*(s,a)$ can be represented as a lookup table (commonly referred to as a Q-table \cite{Arulkumaran2017Deep}) which is computed via iterative update rules.
However, as the dimensionality of the state or action space grows, or when the state or action space becomes continuous, maintaining a Q-table becomes computationally infeasible. To address this limitation, DNN can be employed to approximate the Q-table, such that $Q(s,a) \approx \widetilde{Q}(s,a;\bm{\theta})$, where $\bm{\theta}$ denotes the learnable weights of the DNN. This DNN-based approximation of the action-value function is known as a Deep Q-Network (DQN), which extends RL to handle high-dimensional, continuous state and discrete action spaces \cite{ mnih2015human}.

The trajectory segment $\langle S_t=s,A_t=a,R_{t+1}=R(s,a),S_{t+1}=s^{\prime}\rangle$ forms an “experience sample" for DQN training, which underlies the temporal-difference (TD) learning paradigm \cite{sutton1998reinforcement}. Based on the TD target principle, the DQN training loss function is defined as \cite{mnih2015human}
\begin{equation}
    \mathcal{L}(\boldsymbol{\theta}) = \mathbb{E} \left[ \left( T_{\mathrm{DQN}} - \widetilde{Q}(s, a; \boldsymbol{\theta}) \right)^2 \right],
\end{equation}
where $T_{\mathrm{DQN}}$ is the target Q-value, defined as $T_{\mathrm{DQN}}=R(s,a)+\gamma \max_{a^{\prime}} \widetilde{Q}(s^{\prime},a^{\prime};\boldsymbol{\theta})$. 

\subsection{Approximating the Optimal Control Law for Impedance Tuning  via Double Deep Q-Network}
In this work, we adopt Double Deep Q-Network (DDQN) \cite{van2016deep} rather than standard DQN, as the latter suffers from Q-value overestimation bias that degrades the stability and performance of the RL agent. The core idea of DDQN is to decouple the selection of the optimal action from the estimation of its value by using two separate neural networks: an online network $\widetilde Q(s,a;\boldsymbol{\theta})$ for action selection, and a target network $\hat{Q}(s, a; \boldsymbol{\theta}^-)$ for value estimation \cite{van2016deep}. The target Q-value in DDQN is redefined as
\begin{equation}
    \label{eq:T_DDQN}
    T_{\mathrm{DDQN}} = R(s,a) + \gamma \hat{Q} \left( s', \arg\max_{a'} \widetilde{Q}(s', a'; \boldsymbol{\theta}); \boldsymbol{\theta}^- \right).
\end{equation}

From Section \ref{subsec:Optimal_Control_Law_Tuning}, the optimal control law for the impedance tuning system (i.e., the optimal capacitance adjustment at each step) is the action that maximizes the optimal action-value function of RL. 
This equivalence establishes a direct theoretical foundation for solving the impedance tuning control law using RL. The core of the DDQN algorithm lies in the approximate learning of the optimal action-value function $Q^*(s,a)$. After training, the well-trained DDQN agent can directly output the optimal control action $a^*$ at each step, thereby realizing the optimal control law for impedance tuning as follows
\begin{equation}
    a^*=\arg \max _{a \in \mathcal{A}} \widetilde{Q}(s, a; \bm{\theta}).
\end{equation}

It is worth noting that the DDQN algorithm is inherently designed for continuous state spaces and discrete action spaces. Consequently, the continuous optimal control input (i.e., the optimal action) must be discretized into a finite set of candidate actions before being fed into the DDQN agent. This discretization step introduces an inherent action quantization error into the learned action-value function $\widetilde{Q}(s, a; \bm{\theta})$, which is a fundamental characteristic of the discrete-action RL framework adopted in this work.

This DRL-based implementation paradigm decouples the computationally intensive training phase from the lightweight online inference phase. During online tuning, the agent only performs a single forward pass to select the optimal action, eliminating the need for iterative optimization from scratch, which is crucial for impedance matching applications.

\subsection{Implementation of DRL-Based Impedance Tuning Method}
\label{subsec:algorithm_RL_intro}
To leverage DRL for impedance tuning, we elaborate on the core design of the DRL framework, including the agent, environment, state, action, and reward function, as follows.

1) \textit{Agent:} The agent is the adaptive antenna tuning module, which incorporates a DNN-based RL policy. It autonomously interacts with the operating environment to dynamically adjust the matching network.

2) \textit{Environment:} The environment refers to the dynamic system with which the agent interacts, encompassing the TMN, the source, and the variable load.

3) \textit{State:} Since the magnitude and phase of $\Gamma_\text{in}$ can be measured via a bi-directional coupler and impedance sensor, we adopt them as state variables instead of the real and imaginary parts used in the theoretical analysis. To satisfy the Markov property, the current parameter values of the TMN are also incorporated into the state space. Additionally, the frequency is included as a state variable to support multi-frequency, multi-load impedance tuning scenarios. Thus, we define the state as \begin{equation*}
    s = [|\Gamma_\text{in}|, \sin\phi, \cos \phi, C_p, C_s,f]^T,
\end{equation*}
where $\phi$ denotes the phase of $\Gamma_\text{in}$. Sine and cosine values of phase $\phi$ are used in place of $\phi$ itself to eliminate state discontinuity induced by phase periodicity.    

4) \textit{Action:} Actions correspond to the adjustment increments of capacitors. Given that the action space of the DDQN architecture is discrete, the capacitance adjustment is implemented in a unit-step manner, and thus the action is defined as
\begin{equation}
    a =  [\Delta C_p, \Delta C_s]^T,
\end{equation}
where $a \neq 0$, $\Delta C_p,\Delta C_s \in \{-\Delta C, 0, \Delta C\}$, $\Delta C$ denotes the single tuning step size of the tunable capacitor. Removing the action value $a = 0$ prevents stagnation caused by null action input during the tuning process, while the remaining valid actions can satisfy the full-direction adjustment requirements of dual-capacitor tuning.

5) \textit{Reward:} To balance tuning accuracy and efficiency, a piecewise reward function is designed, which is directly constructed based on the $|\Gamma_\text{in}|$. The immediate reward is defined as $R=r_\text{base}+r_\text{imp}+r_\text{fast}$, where $r_\text{base}$, $r_\text{imp}$ and $r_\text{fast}$ denote the base reward, the improvement reward, and the fast convergence reward, respectively. Designed with piecewise $|\Gamma_\text{in}|$ thresholds, this base reward $r_\text{base}$ provides differentiated incentives that strengthen near ideal matching, with its expression given by
\begin{equation}
r_{\text{base}} = 
\begin{cases} 
100, & |\Gamma_\text{in}| < 0.01 \\
80 + 800 \cdot (0.02 - |\Gamma_\text{in}|), & 0.01 \leq |\Gamma_\text{in}| < 0.02 \\
40 + 600 \cdot (0.06 - |\Gamma_\text{in}|), & 0.02 \leq |\Gamma_\text{in}| < 0.06 \\
-10 -5 \cdot \log_{10}|\Gamma_\text{in}|, & |\Gamma_\text{in}| \geq 0.06
\end{cases}
\end{equation}

Let $\Delta |\Gamma_\text{in}|= |\Gamma_{\text{in},t-1}|-|\Gamma_{\text{in},t}|$, where $|\Gamma_{\text{in},t-1}|$ and $|\Gamma_{\text{in},t}|$ denote the reflection coefficient magnitude before and after the tuning action at time step $t$, respectively. Based on $\Delta |\Gamma_\text{in}|$, this improvement reward term $r_\text{imp}$ is expressed as
\begin{equation}
r_{\text{imp}} = \begin{cases} 
\min\left\{ 30, 300 \cdot \Delta |\Gamma_\text{in}| \right\}, & \Delta |\Gamma_\text{in}| > 0 \\
200 \cdot \Delta |\Gamma_\text{in}|, & \Delta |\Gamma_\text{in}| < -0.02 \\
-0.5, & -0.02 \leq \Delta |\Gamma_\text{in}| \leq 0 
\end{cases}
\end{equation}

The fast reward term $r_\text{fast}$ incentivizes the agent to improve tuning efficiency via step constraints, with its expression
\begin{equation}r_{\text{fast}} = \begin{cases} 
0.1 \cdot (200 - k_{\text{step}}),  & |\Gamma_\text{in}| < 0.01 \&  k_{\text{step}} < 200 \\
0, & \text{else }
\end{cases}
\end{equation}
where $k_{\text{step}}$ is the number of tuning steps required by the agent.

Building upon the detailed design of the above key elements and the fundamentals of DRL, we employ Algorithm \ref{algorithm_RL} to maximize the expected return. The core technical details underpinning Algorithm \ref{algorithm_RL} are elaborated below:
\begin{algorithm}
    \caption{DRL-based Impedance Tuning Method}\label{algorithm_RL}
        \KwIn {Load-frequency training/test pools $\mathcal{P}_{\text{train}}/\mathcal{P}_{\text{test}}$, termination threshold $\varepsilon$, capacitance tuning range $C = [C_{p,\text{min}}, C_{p,\text{max}}] \times [C_{s,\text{min}}, C_{s,\text{max}}]$.}
        \KwOut {Optimal matching solution $\mathcal{C}^* = (C_p, C_s)^*$.}

        \textbf{Initialization}: Initialize environment $\mathcal{E}$; Initialize FIFO replay memory $\mathcal{M}$ with capacity $N_m$;
          Initialize DQN weights $\bm{\theta}$; Set target DQN weights $\bm{\theta}^- = \bm{\theta}$; Set maximum episode $N_{\text{ep}}$,  maximum time step per episode $T_{\text{max}}$, target Q-net update frequency $T_{\text{net}}$, experience batch size $N_e$, $\epsilon$-greedy exploration factor.

        \For{episode $e = 1$ \KwTo $N_{\text{ep}}$}{
             Reset environment $\mathcal{E}$:
             \\ Randomly sample a pair $\{Z_L, f\}$ from $\mathcal{P}_{\text{train}}$; 
             \\ Reset capacitances to fixed initial values $C_p^{(0)}, C_s^{(0)}$;
             \\ Obtain initial state $s_1$; 

            \For{time step $t = 1$ \KwTo $T_{\text{max}}$}{
                Input state $s_t$ to DQN, obtain state-action values $\widetilde{Q}(s_t, a; \bm{\theta}),\ a \in \mathcal{A}$.
                \\ Select action $a_t$ via $\epsilon$-greedy policy based on $\widetilde{Q}(s_t, a; \bm{\theta})$.
                \\ Execute action $a_t$, receive reward $r_{t+1}$, compute next state $s_{t+1}$.
                \\ Store experience tuple $\langle s_t,a_t,r_{t+1},s_{t+1} \rangle$ into replay memory $\mathcal{M}$.\\
                \If{$|\mathcal{M}| \geq N_e$}{
                    Randomly sample a mini-batch of $N_e$ tuples $\langle s_i, a_i, r_{i+1}, s_{i+1} \rangle$ from $\mathcal{M}$.
                    \\ Compute target Q-values $T_{\text{DQN},i}$ for the mini-batch via Eq. \eqref{eq:T_DDQN}.
                    \\ Update DQN weights $\bm{\theta}$ with input $\{s_i\}$ and target $\{T_{\text{DQN},i}\}$.\\
                    \If{$t \mod T_{\text{net}} = 0$}{
                         Update target DQN: $\bm{\theta}^- = \bm{\theta}$.}
                }
                \If{$|\Gamma^{(t)}_{\text{in}}| \leq \varepsilon$}{\textbf{break}.}
                Update time step: $t = t + 1$.
            }
        }
\KwRet Optimal tuning solution $\mathcal{C}^* = (C_p, C_s)^*$.
\end{algorithm}

In contrast to the DQN framework, which uses the same network $\widetilde Q(s,a;\boldsymbol{\theta})$ parameterized by weights $\boldsymbol{\theta}$ to both estimate and target action values, our approach employs a dedicated target network $\widetilde Q(s,a;\boldsymbol{\theta}^-)$ parameterized by weights $\boldsymbol{\theta}^-$ to compute target values. The target network weights are synchronized with the training network weights $\boldsymbol{\theta}$ at intervals of $T_\text{Net}$ time steps. The detailed network architecture is illustrated in Fig. \ref{fig:Q-net}. Specifically, the Q-network is a fully connected architecture, equipped with Dropout regularization to enhance generalization across multi-frequency and multi-load impedance tuning scenarios. The ReLU activation function is utilized in hidden layers to introduce non-linearity, while the output layer remains linear to preserve the range of action value estimates.

\begin{figure}
\centering
\includegraphics[width=3in]{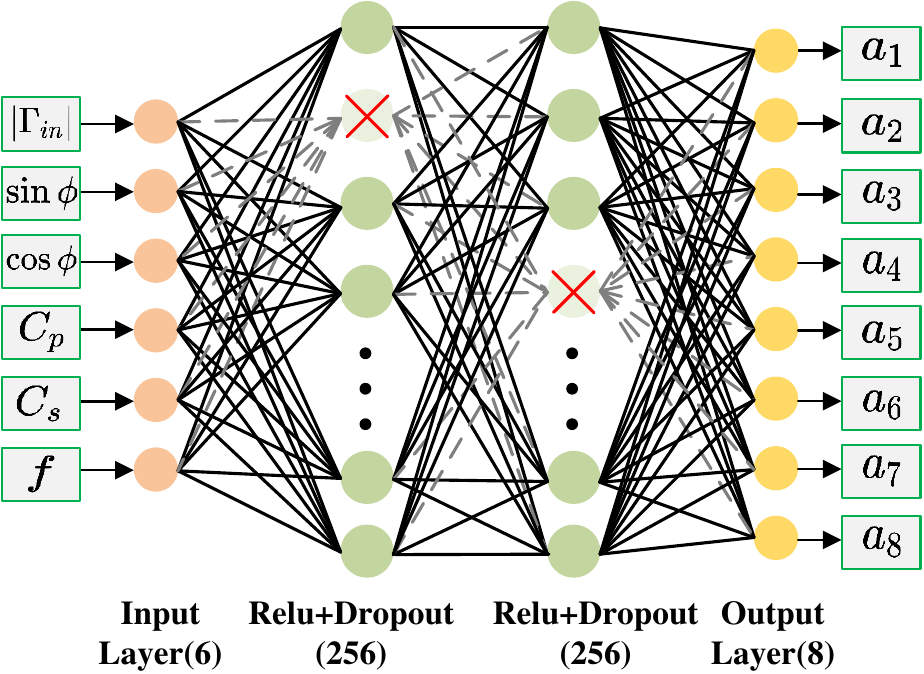}
\caption{Structure of the fully connected Q-network for impedance tuning. It comprises two hidden layers of 256 neurons, each followed by ReLU activation, and Dropout regularization (dropout rate $p=0.2$). The dashed connections with red crosses illustrate the random neuron dropout mechanism during training.
}
\label{fig:Q-net}
\end{figure}

To prevent the agent from converging to a sub-optimal policy due to insufficient exploration, we adopt the $\epsilon$-greedy strategy for decision-making. In this framework, $\epsilon$ represents the probability of performing an exploratory action, where the agent randomly selects from all available actions. Conversely, $1-\epsilon$ denotes the probability of exploiting existing knowledge, in which the agent selects the action with the highest estimated Q-value from the DDQN. Thus, the $\epsilon$-greedy policy can be expressed as 
\begin{equation}
    \pi^{\epsilon} =
\begin{cases}
\pi^*(a \mid s), & \textit{w.p.} 1-\epsilon \\[5pt]
P(a) = \dfrac{1}{|\mathcal{A}|}, & \textit{w.p. } \epsilon
\end{cases}
\end{equation}
where $\pi^*(a|s)$ denotes the greedy policy derived from the Q-network, as previously introduced in Eq. \eqref{eq:policy_expr}. In our implementation, $\epsilon$ is initialized to $\epsilon_0$ to prioritize full exploration in the early stages of tuning, and then undergoes linear decay at a fixed rate of $\epsilon_\text{decay}$ in each time interval. This decay continues until $\epsilon$ reaches a predefined lower bound $\epsilon_{\min}$.

For experience replay, we store $N_e$ recent experience tuples in the replay buffer $\mathcal{M}$ in a “first in, first out” (FIFO) queue structure. This ensures that only the most relevant, up-to-date experiences are retained, with the oldest entry automatically discarded when the buffer reaches capacity. A mini-batch of experience samples is then randomly fetched from $\mathcal{M}$  to train the DQN, which helps break temporal correlations in the training data.

\subsection{Summarizing the Work Flow of DRL-based Impedance Tuning Method}
In this subsection, we summarize the work flow of the proposed impedance tuning method based on DRL. 

\begin{figure}[!t]
\centering
\includegraphics[scale=0.45]{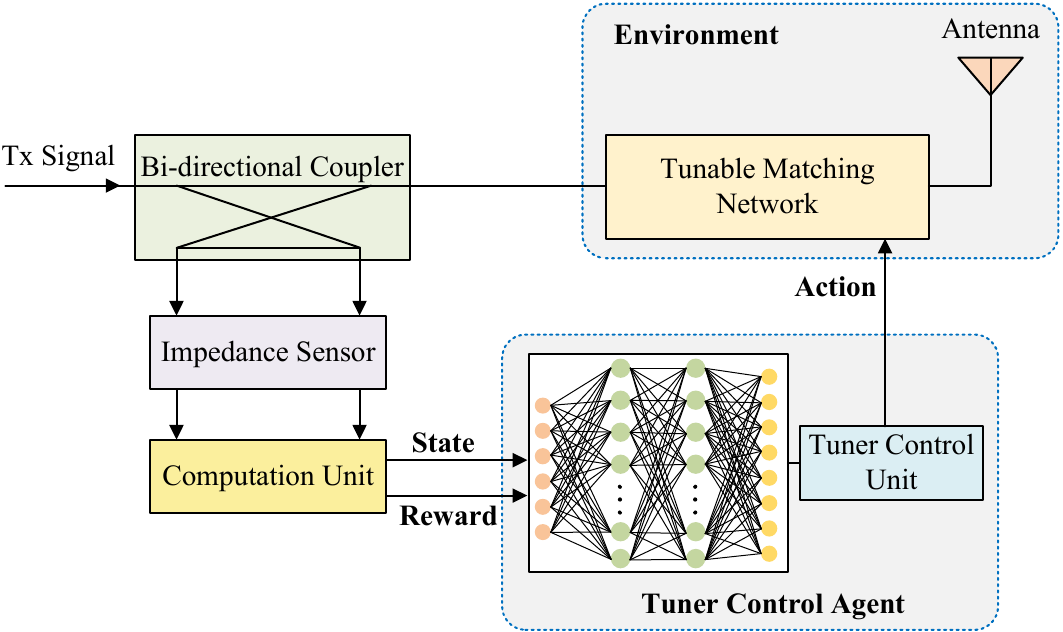}
\caption{Block diagram of an adaptive impedance matching system based DRL.}\label{fig:tuner_system} 
\end{figure}

As shown in Fig. \ref{fig:tuner_system}, in a specific tuning step, the bi-directional coupler and impedance sensor first measure the real-time reflection coefficient $\Gamma_\text{in}$ of the circuit. Subsequently, the computation unit integrates this measured reflection coefficient, along with the current component parameters of the TMN and the operating frequency, through necessary calculations to generate the input state vector $s_t$ for the tuner control agent. Based on the input state, the DDQN selects the optimal discrete action $a_t^*$ according to the learned policy $\pi^*(a|s)$, i.e., the capacitance adjustment command for the two tunable capacitors in the TMN. 
In response to the action command, the tuner control unit directly executes the capacitance adjustment according to $a_t^*$, and then the TMN updates its parameter state. The new state $s_{t+1}$ and the corresponding reward value $r_{t+1}$ (evaluated by matching performance metrics) are fed back to the RL agent to form a closed-loop tuning interaction. During the training phase, the agent collects a large number of state-transition samples $\{s_t,a_t,r_{t+1},s_{t+1}\}$ to optimize the action-value function $\widetilde{Q}(s,a;\bm{\theta})$, the detailed implementation of which is presented in Algorithm \ref{algorithm_RL}. Upon convergence, the trained DDQN model is saved for online deployment.
In the online tuning phase, the pre-trained agent taking the real-time system state $s_t$ as input, directly selects the optimal tuning action $a_t^*$ to adjust the TMN. The updated system state $s_{t+1}$ is then fed back for next online iteration until the target matching accuracy is attained.

\section{Numerical Results and Discussion}
\label{sec:numerical_results}

In this section, we first elaborate on the experimental parameter configurations. Then, we simulate extensive impedance mismatch scenarios to validate the performance of the proposed adaptive impedance matching method.

All experiments are carried out in a Python environment (version 3.9.21). The hardware platform is a workstation equipped with an Intel Xeon Gold 5218 central processing unit (CPU) @ 2.30 GHz and four NVIDIA GeForce RTX 2080 Ti graphics processing units (GPUs). Additionally, the DRL-based adaptive impedance tuning task is formulated as an MDP, and the environment is built upon the Gymnasium framework (version 1.1.1, an upgraded version of OpenAI Gym). The agent’s Q-network is trained with the PyTorch deep learning framework (version 1.13.1), leveraging GPU acceleration (CUDA 11.6).

\subsection{Experimental Setup}
An 8-dimensional discrete action space is designed to enable fixed-step adjustment of the two tunable capacitors $C_p$ and $C_s$. To eliminate dimensional discrepancies among features and ensure training stability, the 6-dimensional state space adopts targeted normalization: $C_p$, $C_s$ and $f$ are globally min-max normalized on the full load-frequency dataset, while the $\sin \phi$ and $\cos \phi$ are not normalized for their inherent range of $[-1,1]$. A multi-stage weighted reward function detailed in Section \ref{subsec:algorithm_RL_intro} guides efficient agent learning. All key experimental parameters of the proposed DRL-based impedance tuning method are summarized in Table \ref{tab:all_RL_params}.

\begin{table}
  \centering
  \caption{Experimental Parameters of the DDQN-Based Adaptive Impedance Tuning Method}
  \label{tab:all_RL_params}
  \begin{tabular}{lc}
    \toprule
    \textbf{Parameter Category} & \textbf{Value} \\
    \midrule
    
    \multicolumn{2}{l}{\textbf{Environment Configuration:}} \\
    \rowcolor{blue!10}
    Capacitance tuning range & $0.5 \sim 21$ pF \\
    Capacitance tuning resolution & $0.5$ pF \\
    \rowcolor{blue!10} Initial capacitance & 11 pF \\ 
    Termination threshold $\varepsilon$ & $0.01$\\
    \rowcolor{blue!10}
    Maximum step per episode $T_{\text{max}}$ & 1000 \\
    Maximum tuning step for test & 200 \\
    
    \addlinespace[0.3em] 
    
    \multicolumn{2}{l}{\textbf{DDQN Architecture:}} \\
    \rowcolor{green!10}
    Network structure & 2 hidden layers
    \\
    Activation function & ReLU \\
    \rowcolor{green!10}
    Regularization & 
    Dropout (0.2) \\
    Optimizer & Adam \\
    \rowcolor{green!10}
    Learning rate & $5 \times 10^{-4}$ \\
    Target network update frequency $T_{\text{net}}$ & 5000 \\
    \addlinespace[0.3em]
    
    \multicolumn{2}{l}{\textbf{Training Protocol:}} \\
    \rowcolor{yellow!10}
    Maximum episode $N_\text{ep}$ & 300 \\
    Experience replay capacity $N_m$ & 50000 \\
    \rowcolor{yellow!10}
    Mini-batch size $N_e$ & 128 \\
    Discount factor $\gamma$ & 0.95 \\
    \rowcolor{yellow!10}
    Initial exploration rate $\epsilon_0$ & 1.0 \\
    Minimum exploration rate $\epsilon_{\min}$ & 0.05 \\
    \rowcolor{yellow!10}
    Exploration rate decay $\epsilon_\text{decay}$ & $1 \times 10^{-5}$ \\
    
    \bottomrule
  \end{tabular}
\end{table}

The source impedance is fixed at 50 $\Omega$. The optimal tuning capacitances $C_p^*$ and $C_s^*$ are pre-defined in the interval of 1 pF to 21 pF with a discrete step of 0.5 pF. The operating frequency $f$ is discretized from 1 GHz to 2 GHz with a step of 0.02 GHz, yielding discrete frequency points. For each combination of $C_p^*$, $C_s^*$ and frequency $f$, the corresponding mismatched load impedance $Z_L=R_L+jX_L$ is calculated via the conjugate matching
equations Eq. \eqref{eq:conjugate matching}. The generated mismatched load impedance and their corresponding frequencies are combined to form a load-frequency sample pool as dataset. As shown in Fig. \ref{mis_imp}, the 81,600 simulated mismatched load
impedance deviates significantly from 50 $\Omega$.

\begin{figure}
\centering
\includegraphics[width=0.5\linewidth]{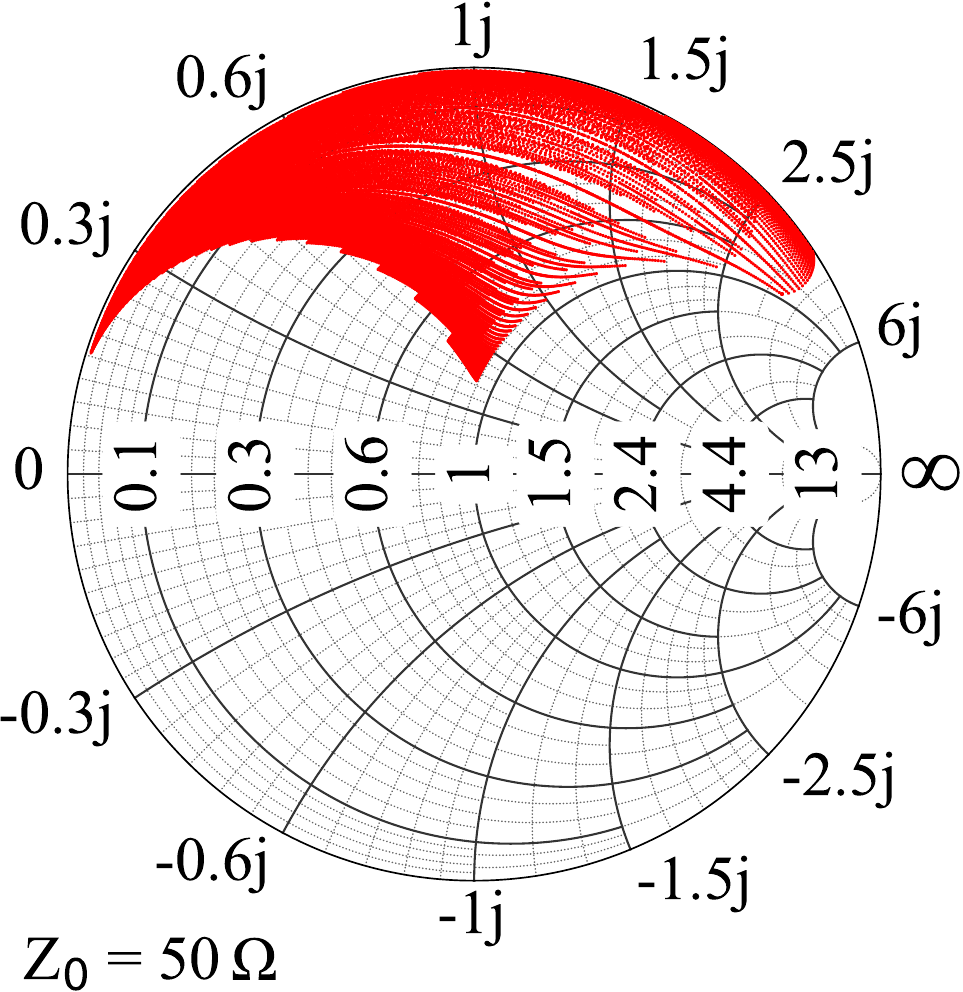}
\caption{Distribution of mismatched impedance under 81,600 simulated mismatched scenarios over the frequency range from 1 GHz to 2 GHz.}\label{mis_imp} 
\end{figure}

To ensure that the training and testing datasets follow the same distribution, we partition the data using a frequency-stratified sampling strategy. Specifically, all samples were first grouped by their operating frequency. Then, within each frequency group, the complete set of mismatched load impedance samples is randomly divided into a training set (60\%) and a
testing set (40\%).
This approach guarantees that the training and testing sets collectively cover the entire frequency spectrum and the full distribution of load. For the loss function, we employ the mean squared error (MSE) to train the Q-network, with its definition given by
\begin{equation}
    \mathcal{L}(\bm{\theta}) = \frac{1}{B} \sum_{(s,a,r,s') \in \mathcal{B}} \left( \widetilde{Q}(s, a; \bm{\theta}) - y_i \right)^2,
\end{equation}
where $\mathcal{B}$ denotes the mini-batch of sampled experience tuples, $B$ is the mini-batch size, $\widetilde Q(s,a; \bm{\theta})$ represents the predicted Q-value of the current Q-network, and $y_i$ is the target Q-value derived from the target Q-network, which is calculated by Eq. \eqref{eq:T_DDQN}. 
Training is performed with PyTorch’s distributed data parallel (DDP) framework on a NVIDIA GeForce RTX 2080 Ti GPU, with a total training time of only 149.99 seconds.

\subsection{Performance of DRL-based Impedance Matching Method}

The impedance tuning agent is trained over a series of episodes, with its training process shown in Fig. \ref{fig:train_process}. The agent exhibits stable convergence after approximately 100 episodes. As shown in Fig. \ref{fig:total_reward}, the cumulative reward per episode initially fluctuates significantly but gradually stabilizes around the zero value after the early training phase, indicating that the agent has learned an effective policy to maximize the cumulative reward. 
Meanwhile, the final reflection coefficient magnitude $|\Gamma_\text{in}|$ (shown in Fig. \ref{fig:T_train}) remains well below the -40 dB (i.e., 0.01) target threshold for most episodes after convergence, with only occasional transient spikes in the early training phase.
These spikes are primarily attributable to the stochastic variation of the load per episode and residual exploration. These results further verify the robustness and reliability of the learned impedance tuning policy. It is worth noting that the agent completes training within only 300 episodes, where one mismatched load is randomly sampled per episode from the training set. This indicates the proposed method yields fast convergence and high sample efficiency, requiring only a small portion of the training dataset.

\begin{figure}
\centering
\subfigure[]{\includegraphics[scale=0.245]{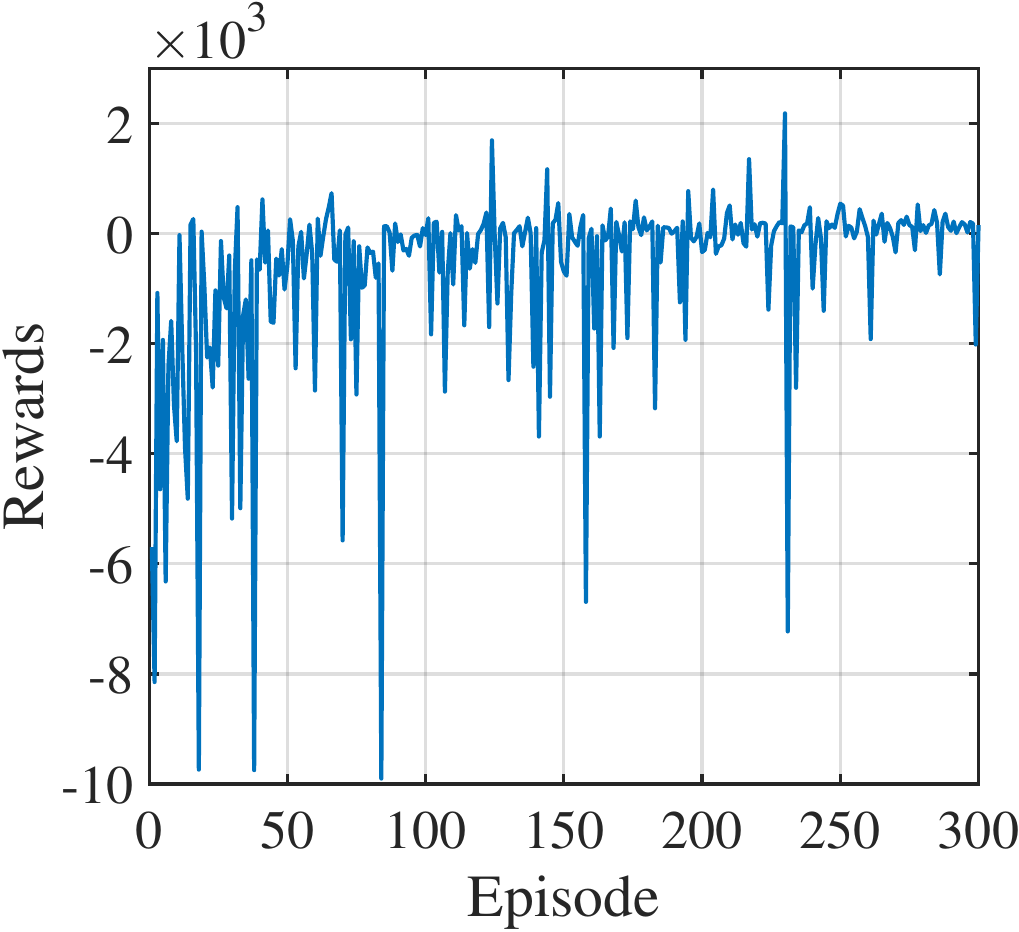}\label{fig:total_reward}}
 \hspace{0.1cm} 
\subfigure[]{\includegraphics[scale=0.245]{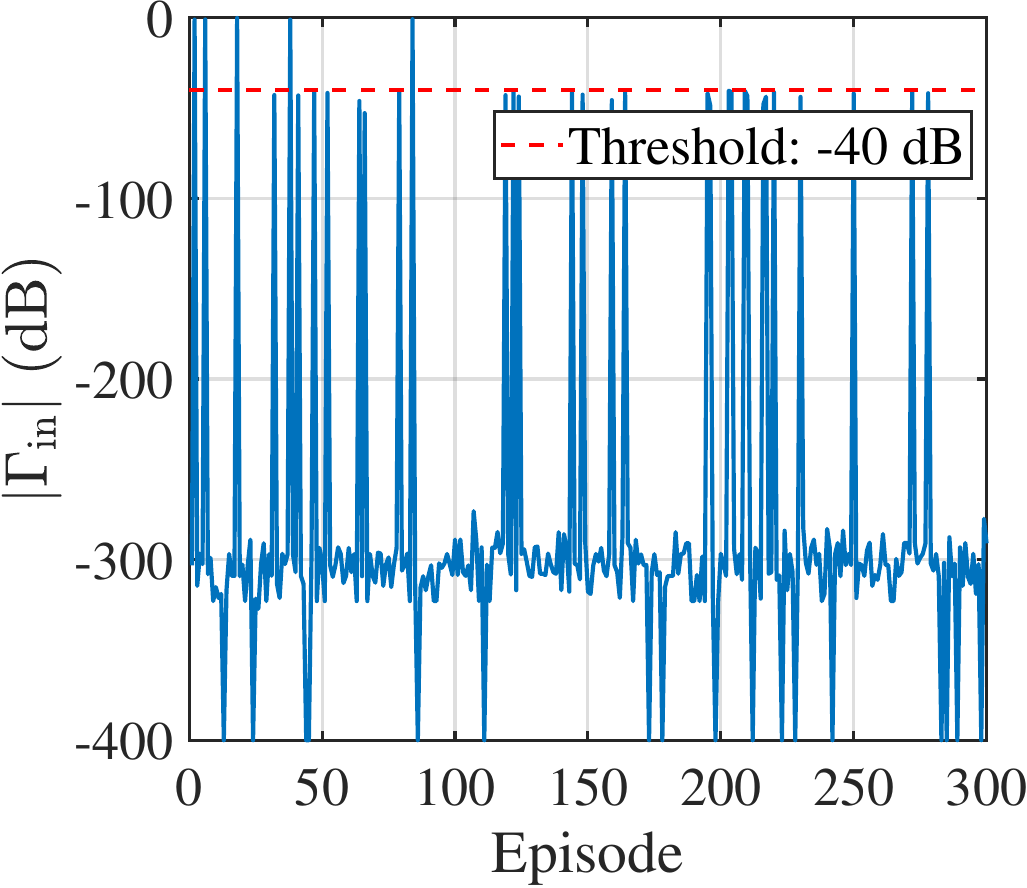} \label{fig:T_train}}
\caption{Training dynamics of the impedance tuning agent. (a) Cumulative reward per episode. (b) Final reflection coefficient magnitude $|\Gamma_\text{in}|$ per episode, with the dashed red line indicating the -40 dB threshold.}
\label{fig:train_process}
\end{figure}

To further evaluate the impedance tuning agent’s performance in adaptive impedance matching, we utilized the test set of 32,640 samples to assess its generalization capability on unseen mismatched scenarios. Baseline methods for comparison include heuristic algorithms (GA \cite{refga} and SAPSO \cite{sapso}), and adaptive moment
estimation with automatic differentiation (AD-Adam) \cite{cheng2025data-driven}. The detailed impedance tuning procedures of SAPSO and AD-Adam are described in \cite{cheng2025data-driven}, and the hyperparameter settings of all three baseline methods are presented in Table \ref{tab:hyperparameters}.
\begin{figure}
\centering
\subfigure[]{\includegraphics[width=0.24\textwidth]{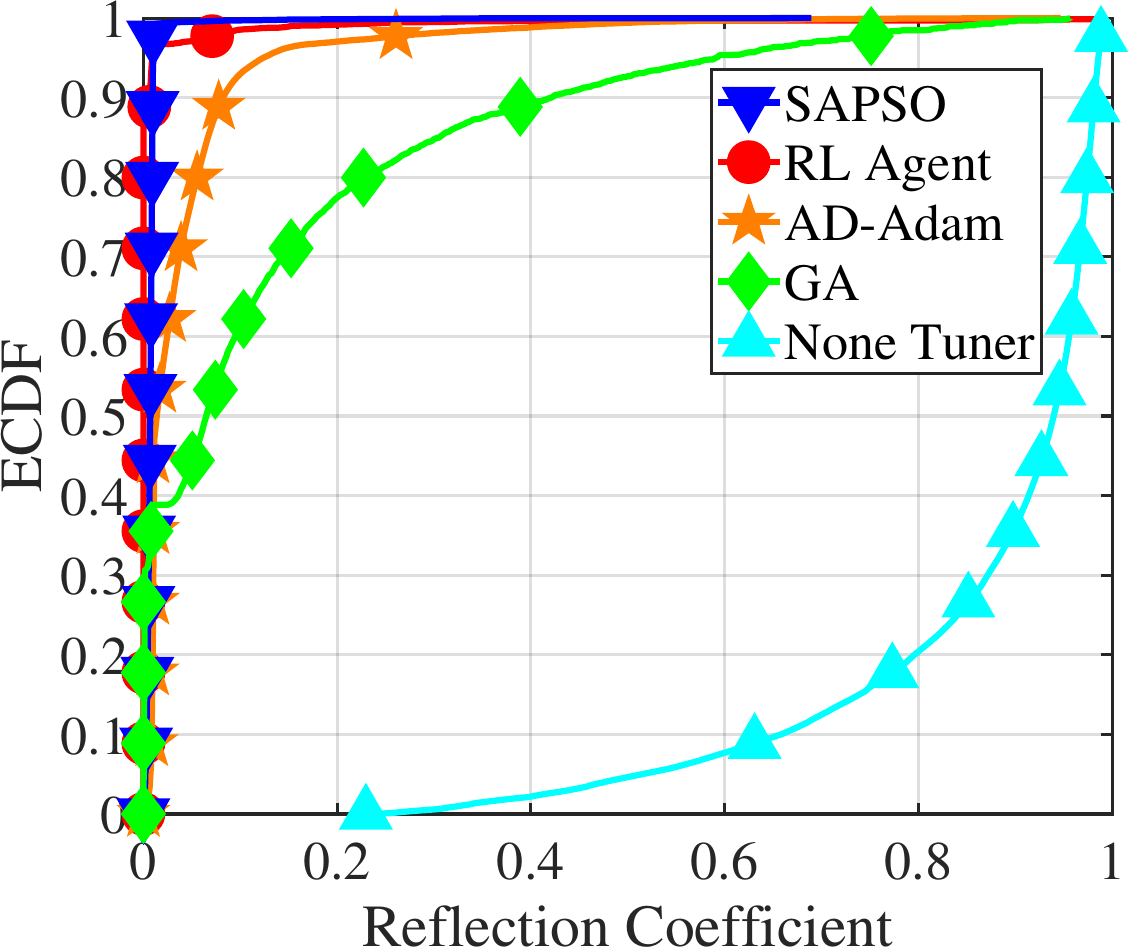}\label{fig:T_CDF_global}}
\hfill  
\subfigure[]{\includegraphics[width=0.24\textwidth]{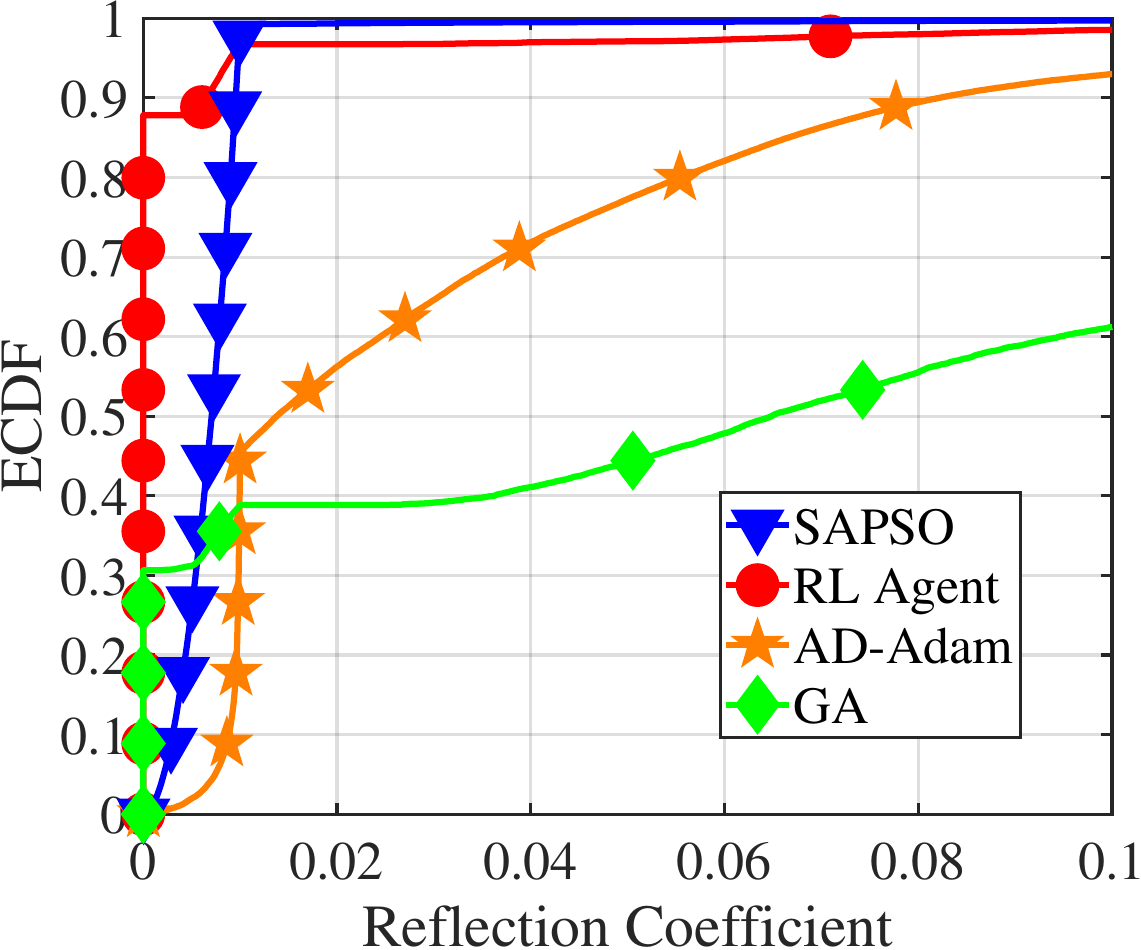}\label{fig:T_CDF_mini}}
\caption{ECDF of the tuned reflection coefficient for the DRL-based impedance tuning agent and baseline methods. (a) Full-range view comparing all methods.
(b) Zoomed-in view for $|\Gamma_\text{in}| \in [0, 0.1]$.}
\label{fig:T_CDF}
\end{figure}
Fig. \ref{fig:T_CDF} presents the empirical cumulative distribution function (ECDF) of the tuned reflection coefficient magnitudes obtained with different matching methods across all test scenarios. In practical engineering applications, a reflection coefficient magnitude below 0.2 is widely adopted as the threshold for high-quality impedance matching \cite{cheng2025data-driven}, corresponding to approximately 96\% of the incident power being delivered to the antenna. Based on this criterion, SAPSO achieves the highest matching accuracy (99.85\% of samples below 0.2), with the proposed DRL-based approach exhibiting slightly inferior yet closely comparable accuracy (99.21\% of samples below 0.2). In comparison, GA delivers lower accuracy (77.45\% of samples below 0.2) than both SAPSO and the proposed method, whereas AD-Adam achieves 97.06\% of samples below the 0.2 threshold. The “None Tuner” case performs poorest, with no samples meeting the 0.2 threshold, highlighting the necessity of the impedance tuner in mismatched scenarios.
\begin{table}
\centering
\caption{Hyperparameter settings for the matching methods SAPSO, AD-Adam and GA}
\label{tab:hyperparameters}
\rowcolors{2}{white}{blue!10}
\begin{tabular}{lccc}
\toprule
\textbf{Parameter} & \textbf{SAPSO} & \textbf{AD-Adam} & \textbf{GA} \\
\midrule
Number of particles & \(20\) & — & — \\
Individual learning factor &$1.5$ & — & — \\
Social learning factor &$1.5$ & — & — \\
Cooling factor & \(0.99\) & — & — \\
Initial capacitances  & — & \(11\,\mathrm{pF}\) & — \\
Learning rate & — & 0.1 & — \\
Exponential decay rates & — & \(0.9,\, 0.999\) & — \\
Stability constant & — & \(10^{-8}\) & — \\
Population size  & — & — & \(20\) \\
Crossover probability  & — & — & \(0.8\) \\
Mutation probability  & — & — & \(0.1\) \\
Maximum iterations & \(200\) & \(200\) & \(200\) \\
Termination threshold & \(0.01\) & \(0.01\) & \(0.01\) \\
\bottomrule
\end{tabular}
\end{table}
To further compare the matching precision of different tuning methods in the high-performance region, Fig. \ref{fig:T_CDF_mini} shows a zoomed-in view of the ECDF curves for $|\Gamma_\text{in}| \leq 0.1$. The ECDF curve of the RL agent rises steeply to a cumulative probability of 96.73\% at the reflection coefficient of 0.01, indicating that the vast majority of its test cases achieve the reflection coefficient below 0.01. In contrast, the ECDF curves of SAPSO and the AD-Adam method rise more gradually, with their cumulative probabilities reaching approximately 99.25\% and 45.58\% at the reflection coefficient of 0.01, respectively. These results confirm that the DRL-based impedance tuning agent achieves competitive matching accuracy compared with SAPSO.

In addition to the ECDFs of the tuned reflection coefficient
magnitude for each matching method, we also summarize the overall mean, median and standard deviation (SD) of the tuned reflection
coefficient magnitudes across all test set. As shown in Table \ref{tab:cross_statistics}, the RL agent and SAPSO achieve mean values well below the 0.01 matching target. The RL agent further yields a median of effectively zero, indicating most test cases achieve perfect matching. In contrast, AD-Adam and GA yield substantially higher means, reflecting inferior overall performance. In terms of stability, SAPSO has the smallest SD well below 0.02, followed closely by the RL agent, while AD-Adam and GA show larger variability. 
\begin{table}
\centering
\caption{Descriptive statistics of the tuned reflection coefficient magnitudes for the four matching methods}
\label{tab:cross_statistics}
\rowcolors{2}{white}{blue!10} 
\begin{tabular}{lccc}
\toprule
\textbf{Method} & \textbf{Mean} & \textbf{Median} & \textbf{SD}\\
\midrule
GA & 0.13680 & 0.06483 & 0.19308 \\
AD-Adam & 0.04027 & 0.01376 & 0.07248 \\
SAPSO& 0.00742 & 0.00706 & 0.01385\\
RL Agent & 0.00718 & 0.00000 & 0.05821\\
\bottomrule
\end{tabular}
\end{table}

To validate the prediction accuracy of optimal TMN component values, Fig. \ref {fig:CsCp_predction_edcf} presents prediction results for optimal capacitances $C_s^*$ and $C_p^*$. As shown in Fig. \ref{fig:CsCp_ecdf}, the proposed DRL-based impedance tuning method enables high prediction precision for both capacitances: the capacitance $C_p^*$ achieves a relative error below 1\% for approximately 97.78\% of samples, and the capacitance $C_s^*$ achieves a relative error below 5\% for approximately 98.77\% of samples. Notably, $C_s^*$ exhibits a slightly higher relative error distribution compared to $C_p^*$, which can be attributed to the stronger coupling between the series capacitance and the load variation, making it more challenging to estimate precisely. These results confirm that our approach maintains accurate prediction of the TMN’s component values, demonstrating its high-performance impedance matching capability.

\begin{figure}[!t]
\centering
\subfigure[]{\includegraphics[scale=0.25]{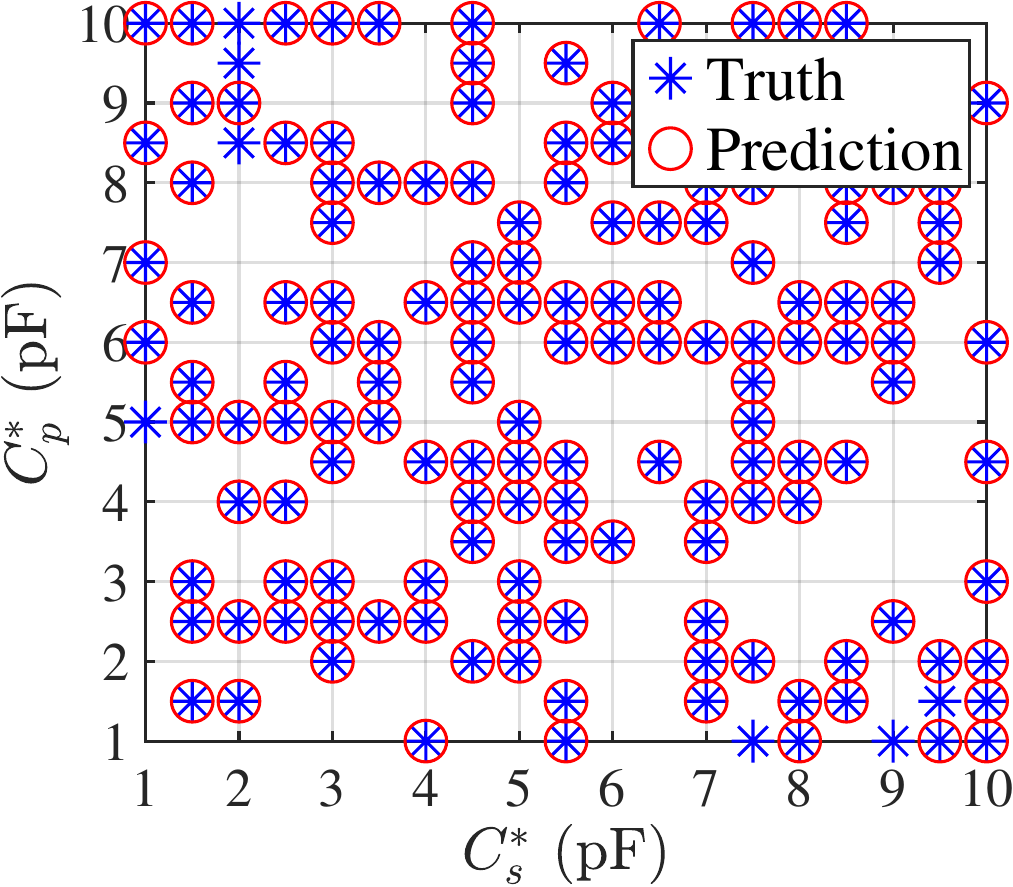}\label{fig:CsCp_predction}}
\hspace{0.1cm}
\subfigure[]{\includegraphics[scale=0.25]{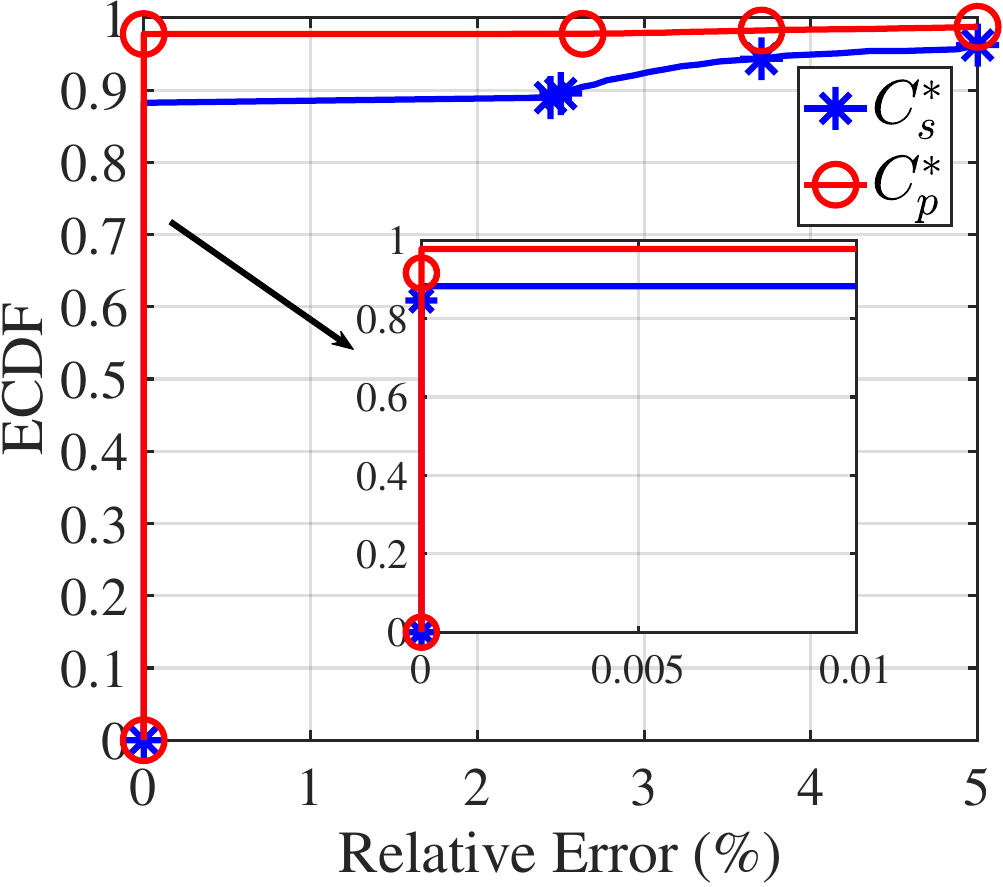}\label{fig:CsCp_ecdf}}
\caption{The matching solution predicted by the proposed DRL-based method. (a) Partial predicted versus true values for optimal capacitances $C_s^*$ and $C_p^*$. (b) ECDF of relative errors for the capacitances $C_s^*$ and $C_p^*$.}
\label{fig:CsCp_predction_edcf}
\end{figure}

In addition to matching accuracy, tuning efficiency is another critical metric for practical impedance matching systems. To evaluate the tuning speed fairly, all impedance tuning methods are executed on the same CPU platform. Note that the DRL-based tuning method is trained on a GPU for online policy learning, while its inference for online tuning is performed on the CPU to ensure a consistent and fair comparison with conventional optimization methods.
Table \ref{tab:tuning_efficiency} presents the tuning efficiency comparison of different impedance tuning methods on the test dataset, including the average tuning steps per test sample, average single-step tuning time, and total execution time.
\begin{table}
\centering
\caption{Tuning Efficiency Metrics of Different Impedance Tuning Methods on the Test Dataset.}
\label{tab:tuning_efficiency}
\rowcolors{2}{white}{blue!10} 
\begin{tabular}{lcccc}
\toprule
\textbf{Metric} & \textbf{AD-Adam} & \textbf{GA} & \textbf{SAPSO}  & \textbf{RL Agent} \\
\midrule
Avg. tuning steps & 165.3 & 129.5 & 24.5 & 21.5 \\
Avg. step time (ms) & 0.76 & 0.44 & 0.82 & 0.33 \\
Execution time (s) & 4099.16 & 1843.00 & 652.26 & 233.50 \\
\bottomrule
\end{tabular}
\end{table}
The RL agent requires only 21.5 average tuning steps per test sample, which is comparable to SAPSO but drastically fewer than those required by AD-Adam and GA. Meanwhile, the RL agent achieves the shortest average single-step time of 0.33 ms, outperforming baseline methods. Consequently, the total execution time of the RL agent is only 233.50 s, which is approximately 17.5 times faster than AD-Adam, nearly 7.9 times faster than GA and approximately 2.8 times faster than SAPSO. The superior tuning efficiency of the DRL-based method originates from its inference-based online tuning mechanism. Once the policy network is trained, it can directly output the optimal tuning action at each step through trained Q-network forward computation, without any iterative optimization or gradient update. In contrast, conventional algorithms must conduct independent and repetitive iterative searches for each individual test sample, which induce substantial redundant computation and execution time.

In summary, the proposed DRL-based impedance tuning method achieves high matching accuracy while exhibiting significantly faster tuning speed than conventional optimization algorithms. These results validate that the DRL-based tuning policy is effective and efficient for impedance matching systems.
\subsection{On the Role of Exploration in Robust Impedance Matching}
As shown in Fig. \ref{fig:T_CDF}, the DRL-based tuning agent achieves excellent impedance matching performance but is slightly outperformed by SAPSO. Table \ref{tab:cross_statistics} further reveals that while their mean reflection coefficients are nearly identical, the RL agent’s SD is approximately four times larger than  that of SAPSO, indicating a few suboptimal tuning cases.
\begin{figure}[!t]
\centering
\subfigure[]{\includegraphics[scale=0.25]{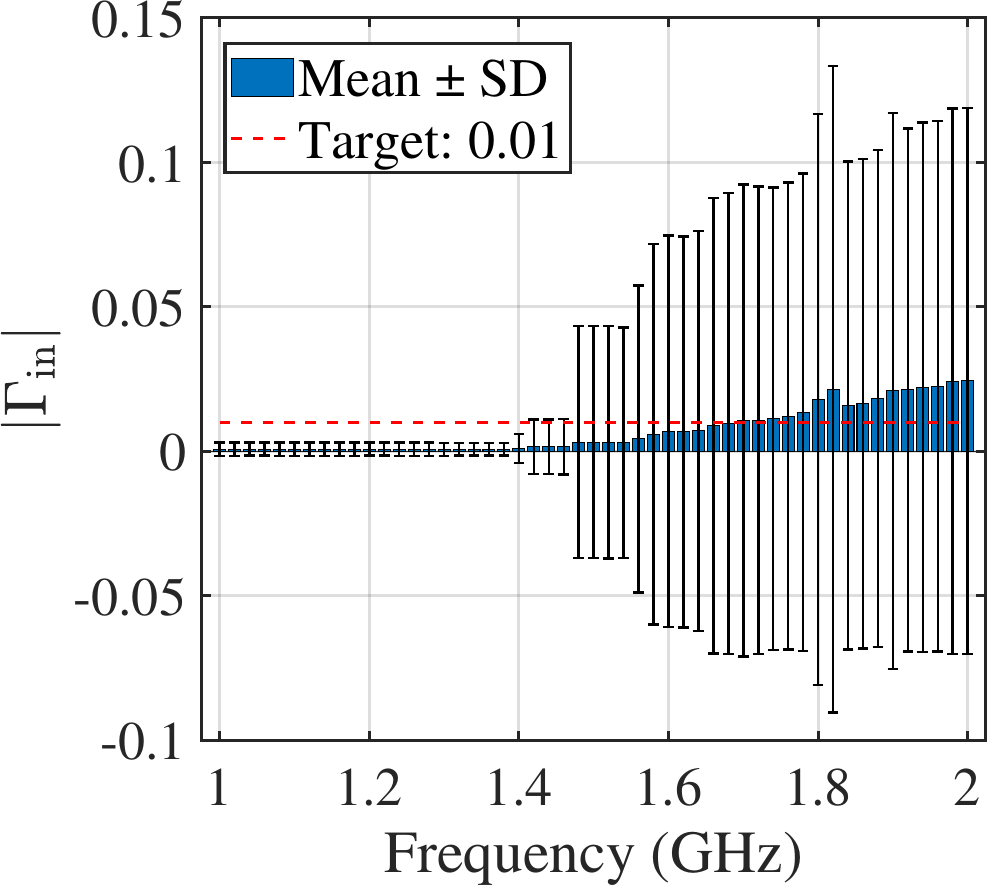}\label{fig:reflec_coeff_vs_freq}}
\hspace{0.1cm}
\subfigure[]{\includegraphics[scale=0.25]{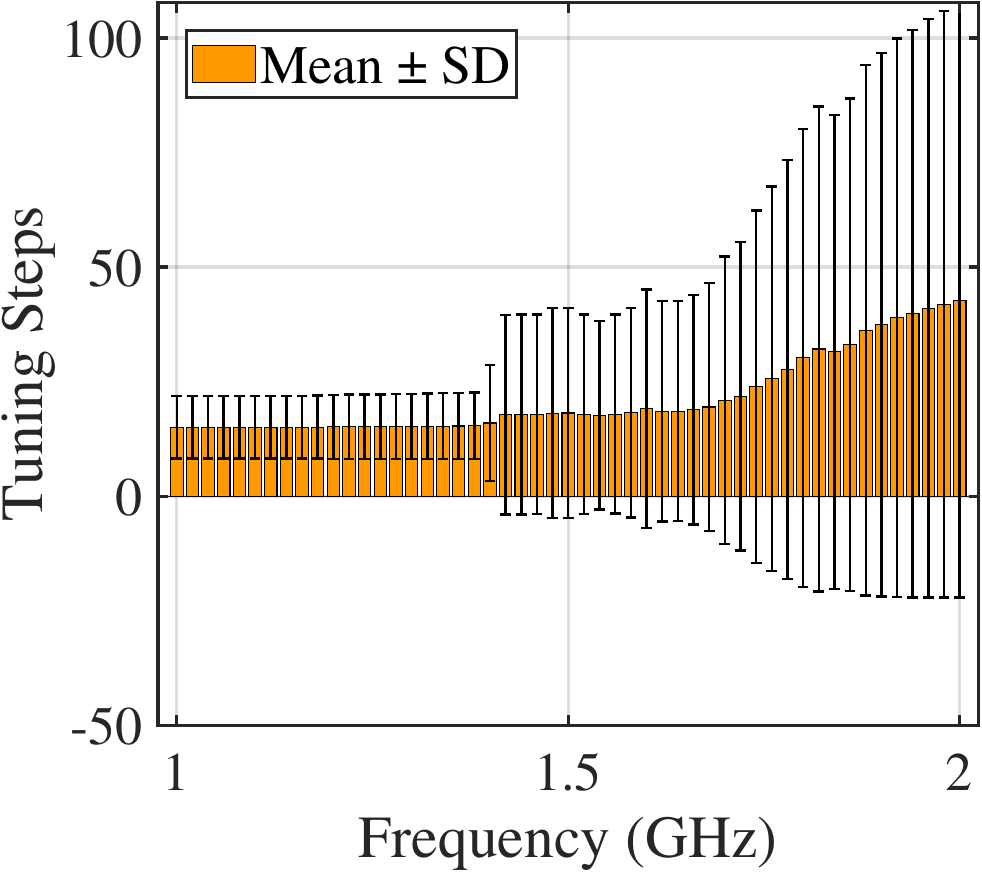}\label{fig:tuning_steps_vs_freq}}
\caption{Frequency-domain performance of the RL agent on test set. (a) Mean $\pm$ SD of relection coefficient $|\Gamma_\text{in}|$, with the dashed line denoting the target threshold. (b) Mean $\pm$ SD of the tuning steps required for impedance matching.}
\label{fig:RL_freq_domain_performance}
\end{figure}
The frequency-domain results shown in Fig. \ref{fig:RL_freq_domain_performance} further reveal that the large SD of the RL agent is primarily attributable to high-frequency variability, while performance remains stable at low frequencies. As shown in Fig. \ref{fig:reflec_coeff_vs_freq}, both the mean and SD of $|\Gamma_\text{in}|$ grow rapidly with frequency, indicating increased matching uncertainty in the high-frequency band. 
Meanwhile, Fig. \ref{fig:tuning_steps_vs_freq} demonstrates that the number of tuning steps also increases markedly and exhibits great fluctuations in the high-frequency region, confirming reduced stability and higher tuning cost at high frequencies.

To elucidate the physical origin of the frequency-dependent performance degradation, Fig. \ref{fig:reflec_coeff_vs_CSCp} illustrates the $|\Gamma_\text{in}|$ surface as a function of $C_s$ and $C_p$. At 1 GHz, the $|\Gamma_\text{in}|$ surface exhibits a single, broad global minimum, forming a convex landscape that enables straightforward convergence. At 2 GHz, the surface becomes markedly steep: the global minimum narrows into a deep valley, accompanied by several secondary local minima. This topological variation directly increases the tuning optimization difficulty at high frequencies. As a result, the RL agent suffers from dramatic fluctuations in $|\Gamma_\text{in}|$ and tuning steps  at high frequencies, consistent with the observed frequency-domain performance.
\begin{figure}
\centering
\subfigure[]{\includegraphics[scale=0.22]{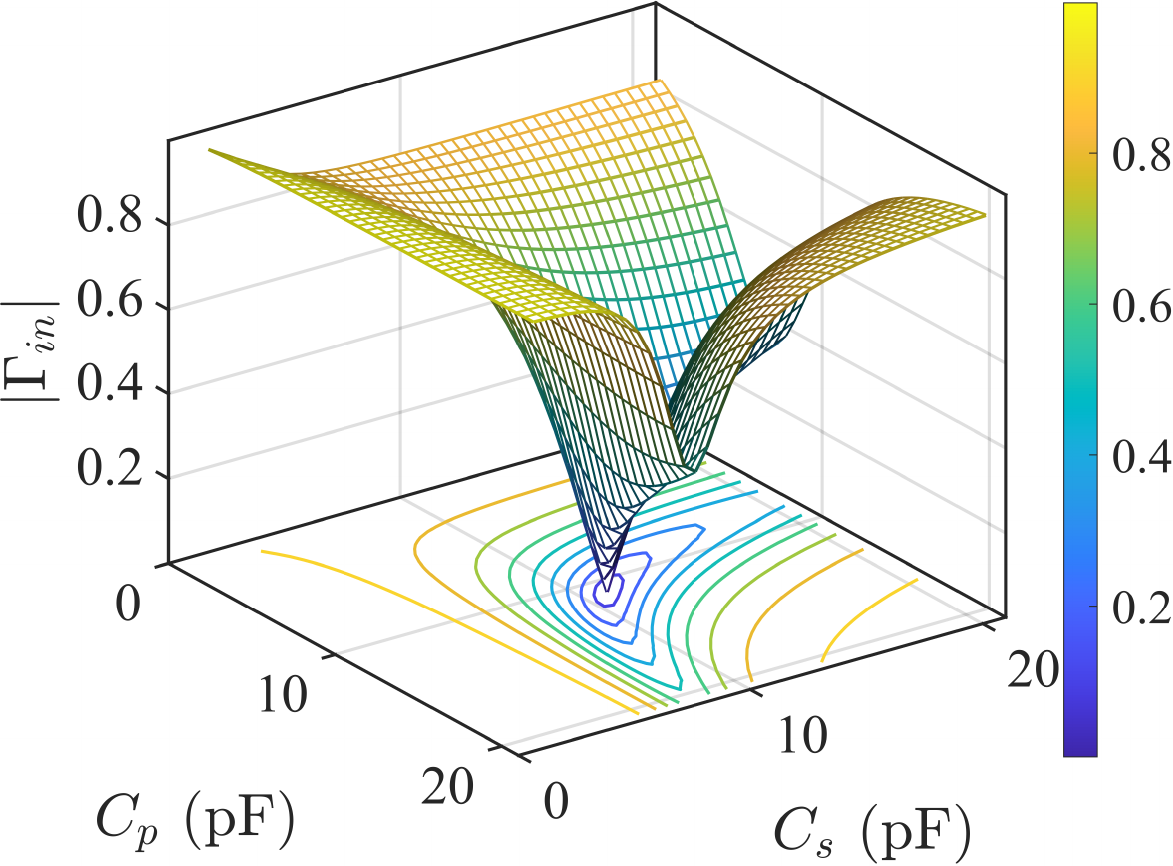}\label{fig:reflec_coeff_vs_CSCp_low_freq}}
\hfill
\subfigure[]{\includegraphics[scale=0.22]{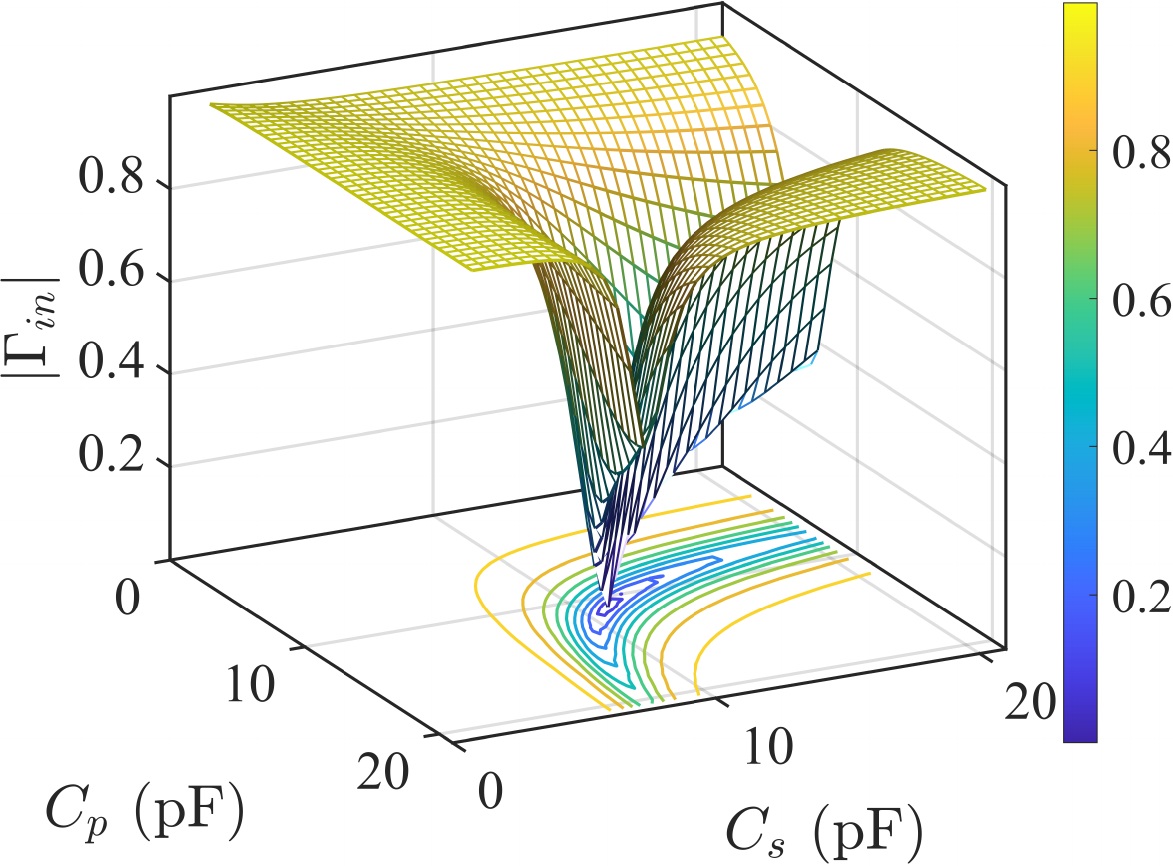}\label{fig:reflec_coeff_vs_CSCp_high_freq}}
\caption{Magnitude of the input reflection coefficient \(|\Gamma_\text{in}|\) as a function of tuning capacitances \(C_s\) and \(C_p\) at (a) 1 GHz and (b) 2 GHz, with the global optimal point at \(C_s^* = C_p^* = 11\ \text{pF}\).}
\label{fig:reflec_coeff_vs_CSCp}
\end{figure}
For both low and high-frequency load samples, the tuning agent proposed in Section \ref{subsec:algorithm_RL_intro} uses the same action space with a fixed tuning step $\Delta C$. This step is suitable for low-frequency scenarios but becomes excessive at high frequencies. Since the impedance and reflection coefficient are more sensitive to capacitance variations at high frequencies, the coarse step tends to induce tuning oscillations and degrade high-frequency performance. Thus, an intuitive improvement is to introduce finer steps (e.g., $\Delta C /2$, $\Delta C /3$) into the original 8-action space to better match the sensitive high-frequency response. However, introducing finer tuning steps extends the action space, which increases training overhead and model complexity. 

To address this issue, this paper introduce a simple yet effective  solution without expanding the action space or introducing extra training overhead. By maintaining a certain action exploration rate during the testing phase, the tuning stability of the agent is improved, thus alleviating oscillation and convergence degradation in high-frequency impedance tuning. To this end, we define the test-phase exploration rate $\epsilon_\text{test}$: during testing, the agent selects a random action with probability $\epsilon_\text{test}$ and the optimal action via the pre-trained Q-network with probability $1-\epsilon_\text{test}$. 
To validate the effectiveness of the proposed test-phase exploration strategy, we conduct experiments under different values of $\epsilon_\text{test}$. Table \ref{tab:cross_statistics_with_epsilon} summarizes the statistics of the tuned reflection coefficient magnitudes for different $\epsilon_\text{test}$ values, with SAPSO as the baseline. It can be observed that as $\epsilon_\text{test}$ increases, both the mean and SD of the reflection coefficient magnitude decrease significantly, indicating improved tuning accuracy and stability of the agent. Notably, for $\epsilon_\text{test} \geq 0.10$, the reduction in SD becomes even more pronounced than that of the mean, enabling the RL agent to outperform SAPSO in both metrics. Additionally, as shown in Fig. \ref{fig:ecdf_explore}, the RL agent with a mere 5\% test-phase exploration achieves superior matching accuracy compared to SAPSO, with 99.6\% of the test samples satisfying $|\Gamma_\text{in}| < 0.01$.
\begin{table}
\centering
\caption{Descriptive statistics of the tuned reflection coefficient magnitudes for different test-phase exploration rates}
\label{tab:cross_statistics_with_epsilon}
\rowcolors{2}{white}{blue!10} 
\begin{tabular}{lcc}
\toprule
\textbf{Method} & \textbf{Mean} & \textbf{SD} \\
\midrule
SAPSO & 0.00742 & 0.01385\\
Agent ($\epsilon_\text{test}=0.00$) & 0.00718 & 0.05821\\
Agent ($\epsilon_\text{test}=0.05$) & 0.00146 & 0.02164\\
Agent ($\epsilon_\text{test}=0.10$) & 0.00088 & 0.01258\\
Agent ($\epsilon_\text{test}=0.20$) & 0.00072 & 0.00601\\
Agent ($\epsilon_\text{test}=0.30$) & 0.00067 & 0.00220\\
\bottomrule
\end{tabular}
\end{table}

\begin{figure}
\centering
\includegraphics[width=0.65\linewidth]{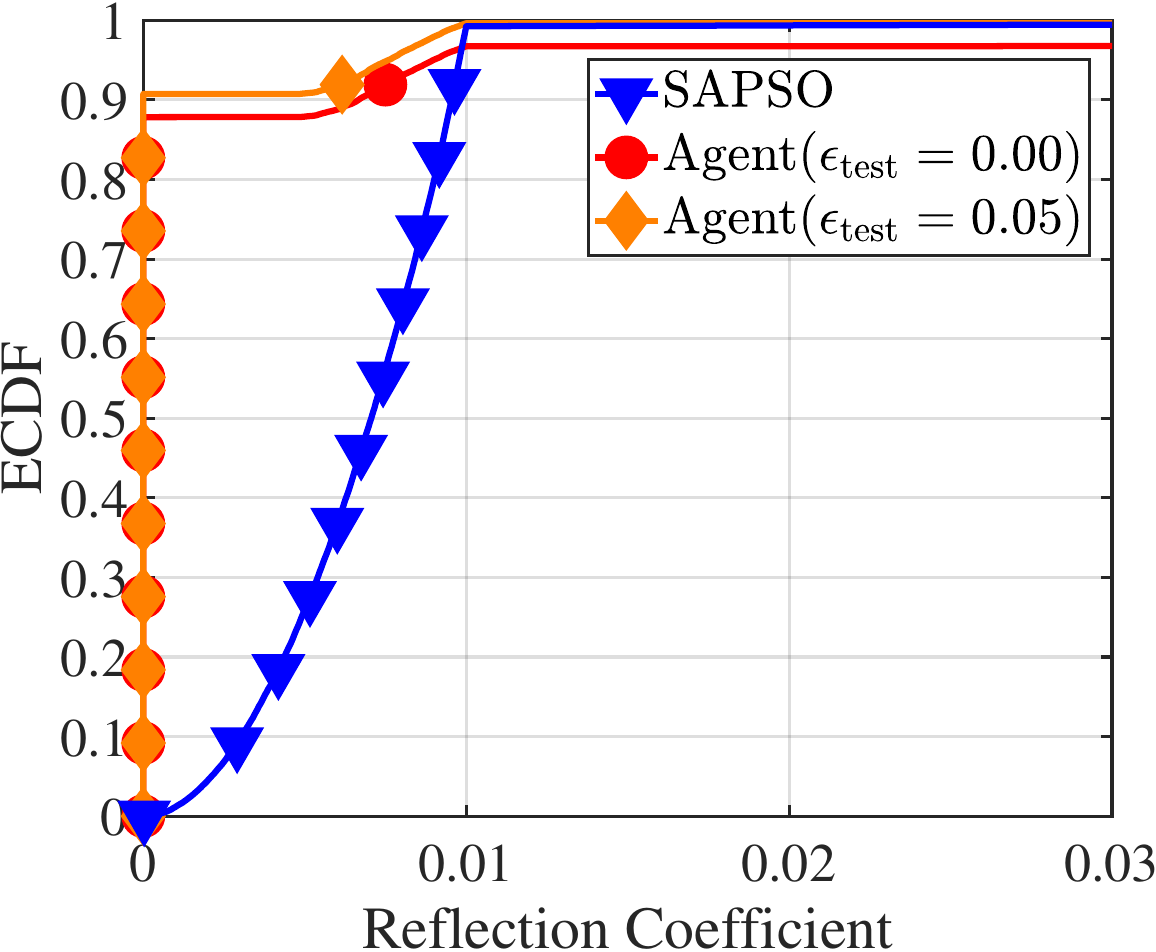}
\caption{ECDF of the tuned reflection coefficient magnitudes for SAPSO and the RL agent with $\epsilon_\text{test}=0.00$ and $\epsilon_\text{test}=0.05 $.}\label{fig:ecdf_explore} 
\end{figure}

To further elaborate the frequency-domain statistical characteristics, Fig. \ref{fig:RL_freq_domain_performance_0.3} presents the detailed results of the RL agent with $\epsilon_\text{test}=0.3$ across the test set. As depicted in Fig. \ref{fig:reflec_coeff_vs_freq_0.3}, the mean value of $|\Gamma_\text{in}|$ remains consistently low (below $1 
\times 10^{-3}$) across the entire frequency band. More importantly, the SD is significantly suppressed below $3 
\times 10^{-3}$ throughout the frequency range. Meanwhile, the tuning steps in Fig. \ref{fig:tuning_steps_vs_freq_0.3} exhibit highly stable behavior with considerably reduced variability.
Compared with the baseline agent ($\epsilon_\text{test}=0$) in Fig. \ref{fig:RL_freq_domain_performance}, Fig. \ref{fig:RL_freq_domain_performance_0.3} shows that the proposed test-phase exploration strategy achieves a remarkable balance between high accuracy and robust stability. These results clearly demonstrate its effectiveness in mitigating the severe oscillations and high variability inherent in high-frequency impedance matching.
\begin{figure}[!t]
\centering
\subfigure[]{\includegraphics[scale=0.25]{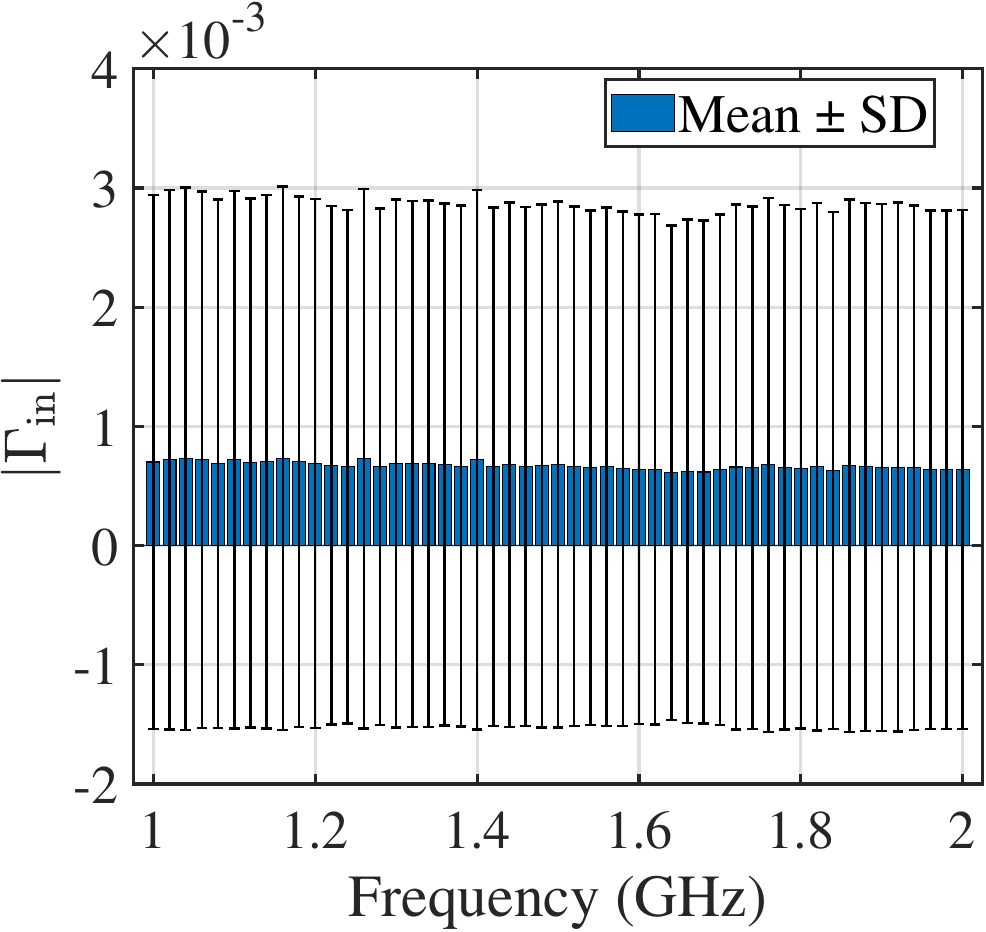}\label{fig:reflec_coeff_vs_freq_0.3}}
\hspace{0.1cm}
\subfigure[]{\includegraphics[scale=0.25]{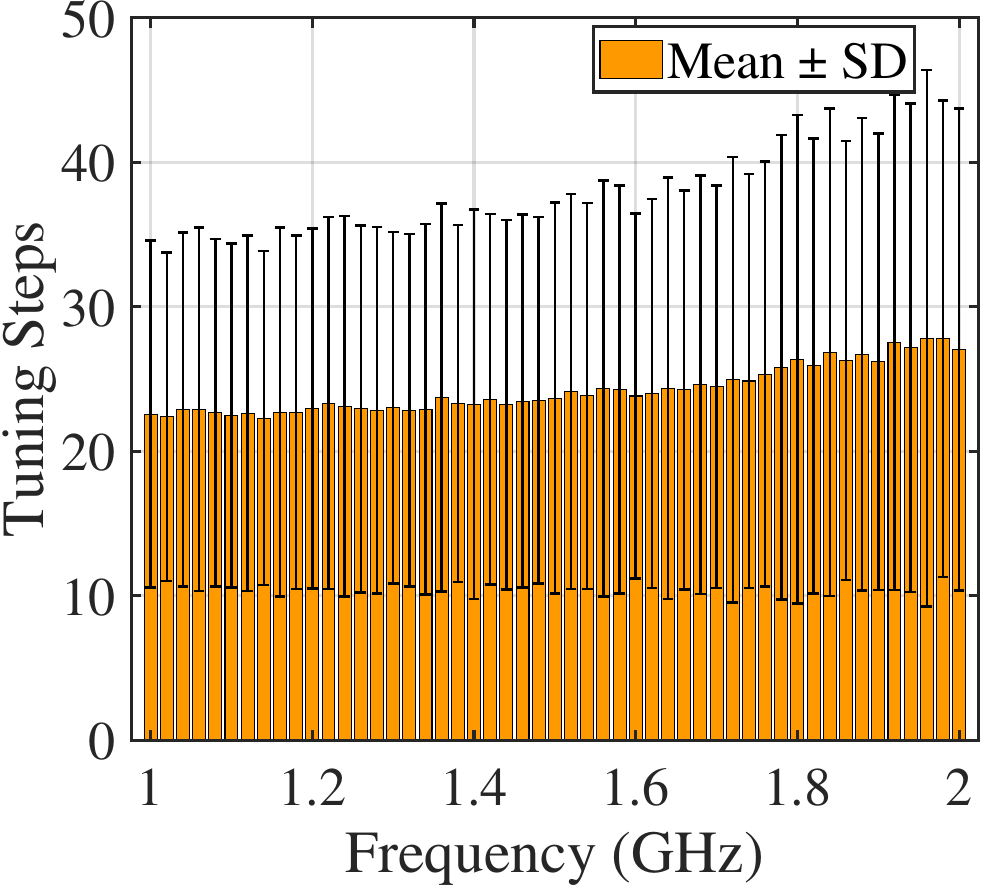}\label{fig:tuning_steps_vs_freq_0.3}}
\caption{Frequency-domain performance of the RL agent with $\epsilon_\text{test} = 0.3$ across the test set.
(a) Mean $\pm$ SD of reflection coefficient $|\Gamma_\text{in}|$.
(b) Mean $\pm$ SD of the tuning steps required for impedance matching.
}
\label{fig:RL_freq_domain_performance_0.3}
\end{figure}
Meanwhile, the test-phase exploration strategy also leads to a substantial reduction in the total execution time required for the agent to complete matching across all test samples. This is primarily attributable to the effective mitigation of high-frequency tuning oscillations, which consume substantial computational time during the matching process. The total execution times of the agent under different test-phase exploration rates, together with their corresponding matching accuracies, are presented in Table \ref{tab:time_acccuracy_explore}.

\begin{table}
\centering
\caption{Execution Time and Matching Accuracy of the RL Agent Under Different Test-Phase Exploration Rates}
\label{tab:time_acccuracy_explore}
\rowcolors{2}{white}{blue!10} 
\begin{tabular}{ccc}
\toprule
  Exploration Rate 
 & Execution Time (s) & $|\Gamma_{\rm in}| < 0.01$ (\%)  \\
\midrule
0.00 & 233.50 & 96.7\\
0.05 & 144.91 & 99.6\\
0.10 & 150.90 & 99.9\\
0.20 & 160.62 & 100.0\\
0.30 & 163.16 & 100.0\\
\bottomrule
\end{tabular}
\end{table}

The performance gain from test-phase exploration stems from the distinct impedance matching solution spaces across frequencies. At low frequencies, where the solution space is smooth, occasional suboptimal actions can be corrected by the agent in subsequent steps with minimal performance degradation. In contrast, at high frequencies, the solution space becomes steeper with numerous local optima, making a deterministic greedy policy prone to trapping the agent in local oscillations and preventing stable convergence. Random action exploration provides an effective mechanism to escape from these local optima, enabling the agent to discover better matching points. 
Therefore, the test-phase exploration strategy significantly enhances the convergence and stability of high-frequency tuning while maintaining the performance of low-frequency tuning.

\section{Conclusion}
\label{sec:conclusion}
In this paper, we have proposed a DRL-based adaptive impedance matching method, achieving significant improvements in tuning accuracy, speed, and stability. First, we have formulated the impedance tuning problem as an optimal control problem, and employed DRL to approximate the optimal control law in a data-driven manner. Then, we have designed a tailored DRL framework for the impedance tuning task, featuring a compact state representation and a piecewise reward function designed specifically for this task. Finally, to mitigate high-frequency tuning variance and oscillations, we have introduced a test-phase exploration mechanism that effectively enhances tuning stability without extra computational overhead. Simulation results have demonstrated that the proposed DRL agent achieves the reflection coefficient below 0.01 for 96.73\% of test samples, outperforming GA and AD-Adam while being competitive with SAPSO in accuracy.
Notably, the proposed agent requires significantly less tuning time than the three baseline methods. 
Furthermore, with a test-phase exploration rate of only 10\%, the agent surpasses SAPSO in terms of tuning accuracy, speed, and stability, achieving a reflection coefficient below 0.01 for 99.9\% of test samples, thereby validating the effectiveness of the proposed matching method.

\IEEEpubidadjcol


\bibliographystyle{IEEEtran}
\bibliography{reference}
\end{document}